\documentclass{amsart}
\usepackage{graphicx} 
\usepackage[english]{babel}
\usepackage[margin=2.5cm]{geometry} 
\usepackage{amssymb, amsfonts, amsthm, amsmath}
\usepackage[hidelinks,pagebackref=true]{hyperref}
\usepackage{caption} 
\usepackage[dvipsnames]{xcolor} 
\colorlet{darkgreen}{green!40!black}
\colorlet{darkviolet}{red!40!blue}
\usepackage{pgf, float} 
\usepackage{tkz-euclide} 
\usetikzlibrary{math,arrows,calc,through, quotes, angles, intersections, patterns,shapes.geometric}
\usepackage[mathscr]{euscript}
\usepackage{accents}
\usepackage{enumitem} 
\usepackage{mathtools} 
\usepackage{parskip}
\usepackage{overpic}

\makeatletter
\newcommand{\SkipTocEntry}[5]{}
\makeatother

\usepackage{thmtools}
\declaretheorem[style=plain,name=Theorem,qed={\tiny$\blacksquare$},numberwithin=section]{theorem}
\declaretheorem[style=plain,name=Corollary,sibling=theorem,qed={\tiny$\blacksquare$}]{corollary}
\declaretheorem[style=definition,name=Definition,sibling=theorem,qed={\tiny$\blacksquare$}]{definition}
\declaretheorem[style=plain,name=Lemma,sibling=theorem,qed={\tiny$\blacksquare$}]{lemma}

\declaretheorem[style=definition,name=Remark,sibling=theorem,qed={\tiny$\blacksquare$}]{remark}

\numberwithin{equation}{section}

\tikzset{bvert/.style={draw,circle,fill=black,minimum size=5pt,inner sep=0pt} }
\tikzset{wvert/.style={draw,circle,fill=white,minimum size=5pt,inner sep=0pt} }
\tikzset{rvert/.style={draw,circle,fill=red,minimum size=5pt,inner sep=0pt} }
\tikzset{blvert/.style={draw,circle,fill=blue,minimum size=5pt,inner sep=0pt} }
\tikzset{redge/.style={draw,densely dashed, red} }

\usetikzlibrary{shapes.geometric}

\newcommand{\C}{\mathbb{C}}
\newcommand{\R}{\mathbb{R}}
\newcommand{\RP}{\mathbb{R}\mathrm{P}}
\newcommand{\Z}{\mathbb{Z}}
\renewcommand{\S}{\mathbb{S}}

\newcommand{\CP}{\mathbb{C}\mathrm{P}}

\DeclareMathOperator{\cro}{cr}
\DeclareMathOperator{\mr}{mr}

\title{Möbius Invariant Dimers and Miquel Dynamics}
\author{Niklas C.\ Affolter \and Luis Mühlhoff}
\date{\today}

\begin{document}

\begin{abstract}
    As is already known, there is a correspondence between circle patterns and the dimer model, which is invariant under Euclidean transformations. We establish a new correspondence that is invariant under the group of \emph{Möbius transformations} $\mathrm{PO(3,1)}$, which is the natural symmetry group of circle patterns. We show that Miquel dynamics preserve the partition function, and that convexity of the t-embedding guarantees positivity of the face weights. We also show that the correspondence has a somewhat hidden symmetry group $\mathrm{PO}(3,3)$, which relates to the Lorentz lift of the t-embedding. Finally, we consider isoradial graphs, Doyle spirals, the once-punctured torus and isothermic circle patterns.
\end{abstract}

\maketitle


\section{Introduction}

A \emph{(bipartite) circle pattern} is a map $p: V(H) \rightarrow \CP^1$, such that $H$ is a planar bipartite graph, and such that for each face $f \in F(H)$ there is a circle $p(f)$ that contains the images of incident vertices.
\emph{Miquel dynamics} is a discrete time dynamics for (bipartite) circle patterns introduced by Richard Kenyon \cite{kenyontalk,ramassamymiquel}, who conjectured that it is in some sense integrable and related to the dimer model. The dynamics are based on Miquel's theorem \cite{miquel}, which allows one to replace a circle in the circle pattern by another circle, accompanied by a local change of combinatorics, see Section~\ref{sec:invariants} for details.

In the following, let us fix a choice of \emph{Kasteleyn orientation} $\kappa: E(H) \rightarrow \S^1 \subset \C$. Let us also denote by $t: F(H) \rightarrow \C$ the \emph{center net} of a circle pattern (also called \emph{conical net} \cite{muellerconical} or \emph{t-realization} \cite{clrmaximal}). Then, \cite{amiquel, klrrdimers} define edge weights
\begin{align}
    \xi(w,b): E(H) \to \R, \qquad (w,b) \mapsto \kappa(w,b)(t(f) - t(f')),
\end{align}
where $(f,f')$ is the dual edge crossing $(w,b)$ from left to right.
To distinguish this choice of weights from the present work, we call $\xi$ the \emph{Euclidean edge weights}. The corresponding Euclidean face weights are real, and the dimer partition function $Z$ is an invariant with respect to Miquel dynamics.
Also, the circle centers $t$ were shown to satisfy a local recurrence equation, which is called \emph{dSKP equation (discrete Schwarzian Kadomtsev–Petviashvili equation)} \cite{ncwqint,dndskp,ksclifford}, and is well-known to be discretely integrable \cite{absoctahedron}. Combined, the results provided the first proof of the integrability and $Z$-invariance of Miquel dynamics.

Circle patterns are invariant under \emph{Möbius transformations}, that is, the action of the group $\mathrm{PO}(3,1)$, in the sense that the Möbius transformation of a circle pattern is a circle pattern. However, Möbius transformations do \emph{not} preserve the dimer partition function given by the Euclidean weights. The main contribution of this article is to introduce a new set of edge weights, such that the dimer partition function is preserved by Miquel dynamics \emph{and} Möbius transformations.

For that purpose, we will associate to $H$ a different bipartite planar graph $G$ (see Section~\ref{sec:dimerscp} for the full setup), that shares its set of white vertices with $G$, and which has only black vertices of degree 3. For a black vertex $b \in B(G)$ we denote the three incident white vertices in counterclockwise order by $w_1,w_2,w_3$, which allows us to define the \emph{Möbius edge weights}
\begin{align}
    \lambda(w_i,b): E(G) \to \R, \qquad (w_i,b) \mapsto \kappa(w_i,b)(p(w_{i+1}) - p(w_{i-1})),
\end{align}
where $\kappa$ is now a Kasteleyn orientation of $G$. We show that the corresponding face weights are real, and if the center net is convex, then they are real positive. Note that the convexity assumption is the same that was used in \cite{klrrdimers}.
We also observe that the dSKP equation appears as a local recurrence equation directly for the intersection points of the circle pattern.

The definitions we give are a generalization of the description in \cite{amiquelmobius}, where the special case of $G = H = \Z^2$ was investigated. Let us emphasize that the generalization to general bipartite graphs (where $G \neq H$) is considerably more involved. We also build upon the framework of \emph{vector-relation configurations} developed in \cite{agprvrc}. One core insight is that there is a correspondence between vector-relation configurations $G \to \CP^1$ with real face weights and circle patterns $H \to \CP^1$.

Although the Euclidean weights are not invariant under Möbius transformations, they were shown to be invariant under a different larger group of transformations of the circle pattern \cite{clrmaximal}: the group of similarity transformations of $\R^{2,2}$, which is isomorphic to $\mathrm{PO}(2,2,1)$. Surprisingly, we show that the Möbius weights are \emph{also} invariant under the same $\mathrm{PO}(2,2,1)$ action. In fact, we show that the Möbius weights are invariant under a $\mathrm{PO}(3,3)$ action, which contains both the Möbius transformations $\mathrm{PO}(3,1)$ and the action of $\mathrm{PO}(2,2,1)$ as subgroups, see Section~\ref{sec:symmetry}.

In the case of Euclidean weights, it was shown in \cite{klrrdimers} that it is possible to use the dimer weights as coordinates on the space of circle patterns defined on the torus and for certain cases in the disk. The corresponding problems for Möbius weights are still open, and provide an interesting direction for further research. The technique used in the disk case involves certain local manipulations of graphs that reduce the combinatorics and geometry. We show these reduction moves also exist in the Möbius case, but we explain why these alone do not suffice in Section~\ref{sec:reductions}.

Finally, in Section~\ref{sec:examples} we discuss special examples of circle patterns and their Möbius weights, which illustrate some additional directions for future research. Specifically, we discuss how \emph{Doyle spirals} \cite{bdsdoyle} have biperiodic Möbius face weights; how a special case of our setup reproduces the shear coordinates of \emph{punctured Teichmüller theory} \cite{pennertm}; how the Möbius face weights of \emph{isoradial} circle patterns \cite{ksisoradial,bdtisosurvey} are expressible in a fashion quite reminiscent of the case of Euclidean face weights \cite{kenyondirac}; and we show how certain \emph{isothermic circle patterns} \cite{muellerconical} can be characterized among circle patterns as those where the Möbius face weights are invariant when exchanging the role of black and white vertices.

\addtocontents{toc}{\protect\SkipTocEntry}
\subsection*{Acknowledgments}

We would like to thank Richard Kenyon and Alexander Bobenko for insightful discussions.

\section{Dimers and circle patterns}

\label{sec:dimerscp}

\subsection{Bicolored triangulations}

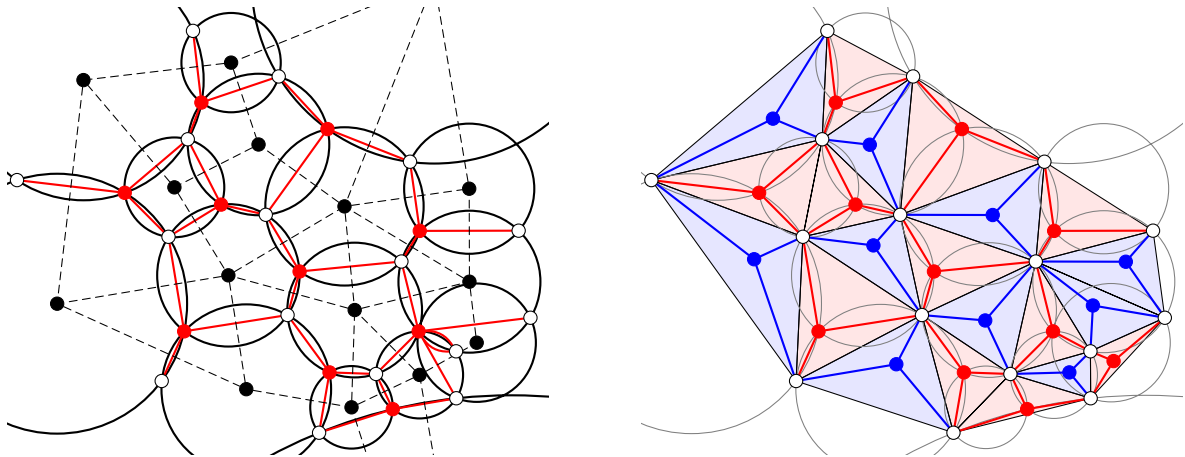
\begin{figure}
    \centering
    \pgfdeclarelayer{vertices}
    \pgfsetlayers{main,vertices}
    \begin{tikzpicture}[line cap=round,line join=round,>=triangle 45,x=1cm,y=1cm, scale=.3]
        \begin{pgfonlayer}{vertices}
            \coordinate[wvert] (A)  at (4.55,4.78);
            \coordinate[wvert] (B)  at (4.78,9.56);
            \coordinate[wvert] (C)  at (8.56,7.55);
            \coordinate[wvert] (D)  at (3.70,.48);
            \coordinate[wvert] (E)  at (7.99,1.45);
            \coordinate[wvert] (F)  at (14.35,3.79);
            \coordinate[wvert] (G)  at (8.95,-2.97);
            \coordinate[wvert] (H)  at (13.97,-.61);
            \coordinate[wvert] (I)  at (19.14,.75);
            \coordinate[wvert] (J)  at (-2.98,2.98);
            \coordinate[wvert] (K)  at (3.40,-5.88);
            \coordinate[wvert] (L)  at (10.34,-8.16);
            \coordinate[wvert] (M)  at (12.85,-5.57);
            \coordinate[wvert] (N)  at (16.38,-4.57);
            \coordinate[wvert] (O)  at (19.66,-3.09);
            \coordinate[wvert] (P)  at (16.39,-6.64);

            \coordinate[bvert,red] (R1)  at (6.02,1.90);
            \coordinate[bvert,red] (R2)  at (10.71,5.23);
            \coordinate[bvert,red] (R3)  at (9.48,-1.05);
            \coordinate[bvert,red] (R4)  at (4.39,-3.69);
            \coordinate[bvert,red] (R5)  at (1.77,2.42);
            \coordinate[bvert,red] (R6)  at (5.15,6.40);
            \coordinate[bvert,red] (R7)  at (14.78,.74);
            \coordinate[bvert,red] (R8)  at (14.73,-3.69);
            \coordinate[bvert,red] (R9)  at (10.80,-5.50);
            \coordinate[bvert,red] (R10) at (13.59,-7.12);

            \coordinate[bvert] (O1)  at (6.34,-1.22);
            \coordinate[bvert] (O2)  at (7.13,-6.24);
            \coordinate[bvert] (O3)  at (7.67,4.55);
            \coordinate[bvert] (O4)  at (11.46,1.82);
            \coordinate[bvert] (O5)  at (11.91,-2.75);
            \coordinate[bvert] (O6)  at (16.97,-1.50);
            \coordinate[bvert] (O7)  at (16.96,2.60);
            \coordinate[bvert] (O8)  at (3.96,2.66);
            \coordinate[bvert] (O9)  at (-1.21,-2.47);
            \coordinate[bvert] (O10) at (-.05,7.39);

            \coordinate[bvert] (O11) at (6.46,8.16);
            \coordinate[] (O12) at (15.42,11.77);
            \coordinate[] (O13) at (18.83,-29.07);
            \coordinate[bvert] (O14) at (14.77,-5.61);
            \coordinate[bvert] (O15) at (17.29,-4.19);
            \coordinate[bvert] (O16) at (11.77,-7.04);
        \end{pgfonlayer}

        \clip (-3.41,-9.10) rectangle (20.45,10.57);

        \draw[thick,black]
            (O1) circle (3.14cm) (O2) circle (3.74cm) (O3) circle (3.12cm) (O4) circle (3.49cm) (O5) circle (2.97cm) (O6) circle (3.13cm) (O7) circle (2.87cm) (O8) circle (2.20cm) (O9) circle (5.73cm) (O10) circle (5.30cm) (O14) circle (1.92cm) (O15) circle (2.61cm) (O16) circle (1.82cm) (O11) circle (2.19cm) (O12) circle (8.06cm) (O13) circle (22.56cm)
        ;
        \draw[thick,red]
            (A) -- (R1) -- (E) (E) -- (R2) (E) -- (R3) -- (G) (R3) -- (H) (H) -- (R7) -- (F) (F) -- (R2) (R7) -- (I) (H) -- (R8) -- (O) (G) -- (R4) -- (D) (D) -- (R1) (D) -- (R5) -- (A) (R4) -- (K) (G) -- (R9) -- (M) (R9) -- (L) (M) -- (R8) (R8) -- (P) (M) -- (R10) -- (L) (R10) -- (P) (J) -- (R5) (A) -- (R6) -- (B) (R6) -- (C) (C) -- (R2)
        ;
        \draw[thick,red]
            (N) edge[bend left] (R8) edge[bend right] (R8)
        ;

        \draw[densely dashed]
            (O12) -- (O11) (O11) -- (O10) (O11) -- (O3) (O3) -- (O8) (O8) -- (O10) (O10) -- (O9) (O9) -- (O2) (O2) -- (O16) (O14) -- (O16) (O14) -- (O5) (O14) -- (O15) (O15) -- (O6) (O6) -- (O7) (O7) -- (O12) (O12) -- (O4) (O4) -- (O3) (O8) -- (O1) (O1) -- (O4) (O4) -- (O7) (O6) -- (O4) (O4) -- (O5) (O5) -- (O1) (O1) -- (O2) (O1) -- (O9) (O16) -- (O5) (O5) -- (O6) (O13) -- (O16) (O13) -- (O14)
        ;

    \end{tikzpicture}
    \hspace{10mm}
    \begin{tikzpicture}[line cap=round,line join=round,>=triangle 45,x=1cm,y=1cm, scale=.3]
        \begin{pgfonlayer}{vertices}
            \coordinate[wvert] (A)  at (4.55,4.78);
            \coordinate[wvert] (B)  at (4.78,9.56);
            \coordinate[wvert] (C)  at (8.56,7.55);
            \coordinate[wvert] (D)  at (3.70,.48);
            \coordinate[wvert] (E)  at (7.99,1.45);
            \coordinate[wvert] (F)  at (14.35,3.79);
            \coordinate[wvert] (G)  at (8.95,-2.97);
            \coordinate[wvert] (H)  at (13.97,-.61);
            \coordinate[wvert] (I)  at (19.14,.75);
            \coordinate[wvert] (J)  at (-2.98,2.98);
            \coordinate[wvert] (K)  at (3.40,-5.88);
            \coordinate[wvert] (L)  at (10.34,-8.16);
            \coordinate[wvert] (M)  at (12.85,-5.57);
            \coordinate[wvert] (N)  at (16.38,-4.57);
            \coordinate[wvert] (O)  at (19.66,-3.09);
            \coordinate[wvert] (P)  at (16.39,-6.64);

            \coordinate[bvert,red] (R1)  at (6.02,1.90);
            \coordinate[bvert,red] (R2)  at (10.71,5.23);
            \coordinate[bvert,red] (R3)  at (9.48,-1.05);
            \coordinate[bvert,red] (R4)  at (4.39,-3.69);
            \coordinate[bvert,red] (R5)  at (1.77,2.42);
            \coordinate[bvert,red] (R6)  at (5.15,6.40);
            \coordinate[bvert,red] (R7)  at (14.78,.74);
            \coordinate[bvert,red] (R8)  at (14.73,-3.69);
            \coordinate[bvert,red] (R9)  at (10.80,-5.50);
            \coordinate[bvert,red] (R10) at (13.59,-7.12);
            \coordinate[bvert,red] (R11)  at (17.4,-5);

            \coordinate[bvert,blue] (Q1)  at (6.65,4.54);
            \coordinate[bvert,blue] (Q2)  at (2.38,5.69);
            \coordinate[bvert,blue] (Q3)  at (1.55,-.51);
            \coordinate[bvert,blue] (Q4)  at (7.82,-5.14);
            \coordinate[bvert,blue] (Q5)  at (11.74,-3.21);
            \coordinate[bvert,blue] (Q6)  at (15.45,-5.50);
            \coordinate[bvert,blue] (Q7)  at (16.49,-2.55);
            \coordinate[bvert,blue] (Q8)  at (17.94,-.62);
            \coordinate[bvert,blue] (Q9)  at (12.07,1.44);
            \coordinate[bvert,blue] (Q10) at (6.81,.10);

            \coordinate[] (O1)  at (6.34,-1.22);
            \coordinate[] (O2)  at (7.13,-6.24);
            \coordinate[] (O3)  at (7.67,4.55);
            \coordinate[] (O4)  at (11.46,1.82);
            \coordinate[] (O5)  at (11.91,-2.75);
            \coordinate[] (O6)  at (16.97,-1.50);
            \coordinate[] (O7)  at (16.96,2.60);
            \coordinate[] (O8)  at (3.96,2.66);
            \coordinate[] (O9)  at (-1.21,-2.47);
            \coordinate[] (O10) at (-.05,7.39);
            \coordinate[] (O11) at (6.46,8.16);
            \coordinate[] (O12) at (15.42,11.77);
            \coordinate[] (O13) at (18.83,-29.07);
            \coordinate[] (O14) at (14.77,-5.61);
            \coordinate[] (O15) at (17.29,-4.19);
            \coordinate[] (O16) at (11.77,-7.04);
        \end{pgfonlayer}

        \clip (-3.41,-9.10) rectangle (20.45,10.57);

        \filldraw[draw=black,fill=red,fill opacity=.1]
            (B.center) -- (A.center) -- (C.center) -- cycle
            (A.center) -- (D.center) -- (E.center) -- cycle
            (C.center) -- (E.center) -- (F.center) -- cycle
            (E.center) -- (G.center) -- (H.center) -- cycle
            (F.center) -- (H.center) -- (I.center) -- cycle
            (A.center) -- (J.center) -- (D.center) -- cycle
            (D.center) -- (K.center) -- (G.center) -- cycle
            (G.center) -- (L.center) -- (M.center) -- cycle
            (H.center) -- (M.center) -- (N.center) -- cycle
            (O.center) -- (N.center) -- (P.center) -- cycle
            (M.center) -- (L.center) -- (P.center) -- cycle
        ;

        \filldraw[draw=black,fill=blue,fill opacity=.1]
            (A.center) -- (B.center) -- (J.center) -- cycle
            (J.center) -- (K.center) -- (D.center) -- cycle
            (K.center) -- (G.center) -- (L.center) -- cycle
            (D.center) -- (E.center) -- (G.center) -- cycle
            (E.center) -- (A.center) -- (C.center) -- cycle
            (E.center) -- (F.center) -- (H.center) -- cycle
            (H.center) -- (G.center) -- (M.center) -- cycle
            (H.center) -- (O.center) -- (I.center) -- cycle
            (H.center) -- (N.center) -- (O.center) -- cycle
            (N.center) -- (M.center) -- (P.center) -- cycle
        ;
        \draw[gray]
            (O1) circle (3.14cm) (O2) circle (3.74cm) (O3) circle (3.12cm) (O4) circle (3.49cm) (O5) circle (2.97cm) (O6) circle (3.13cm) (O7) circle (2.87cm) (O8) circle (2.20cm) (O9) circle (5.73cm) (O10) circle (5.30cm) (O14) circle (1.92cm) (O15) circle (2.61cm) (O16) circle (1.82cm) (O11) circle (2.19cm) (O12) circle (8.06cm) (O13) circle (22.56cm)
        ;
        \draw[thick,red]
            (A) -- (R1) -- (E) (E) -- (R2) (E) -- (R3) -- (G) (R3) -- (H) (H) -- (R7) -- (F) (F) -- (R2) (R7) -- (I) (H) -- (R8) (G) -- (R4) -- (D) (D) -- (R1) (D) -- (R5) -- (A) (R4) -- (K) (G) -- (R9) -- (M) (R9) -- (L) (M) -- (R8) (R8) -- (N) (M) -- (R10) -- (L) (R10) -- (P) (J) -- (R5) (A) -- (R6) -- (B) (R6) -- (C) (C) -- (R2) (R11) -- (P) (N) -- (R11) -- (O)
        ;
        \draw[thick,blue]
            (E) -- (Q1) (Q1) -- (A) (Q1) -- (C) (A) -- (Q2) (Q2) -- (B) (Q2) -- (J) (J) -- (Q3) (Q3) -- (D) (Q3) -- (K) (K) -- (Q4) (Q4) -- (G) (Q4) -- (L) (G) -- (Q5) (Q5) -- (H) (Q5) -- (M) (M) -- (Q6) (Q6) -- (N) (Q6) -- (P) (N) -- (Q7) (Q7) -- (H) (Q7) -- (O) (O) -- (Q8) (Q8) -- (I) (Q8) -- (H) (H) -- (Q9) (Q9) -- (F) (Q9) -- (E) (E) -- (Q10) (Q10) -- (D) (Q10) -- (G)
        ;
    \end{tikzpicture}

    \caption{Left: an example of a bipartite circle pattern (black) with graph $H$ (red) and convex t-embedding (dashed). Right: the corresponding bicolored triangulation $T$ with the dimer graph $G$ (blue). There appears to be a red degree four vertex in the circle pattern (left), but it is actually two red vertices (see right) mapped to the same point.}
    \label{fig:excp}
\end{figure}

Let $T$ be a finite triangulation of the disk, together with a bicoloring of the set of triangles of $T$ into blue and red triangles, see also Figure~\ref{fig:excp}. We generally assume that the degree of every interior vertex of $T$ is at least 2. To a triangulation we associate two other graphs $G,H$ as follows. To construct $G$ from $T$, add a vertex to $G$ for each blue triangle of $T$, and connect this vertex with three edges to the vertices of the triangle. Subsequently, we delete all the edges of $T$. Analogously, we construct $H$ from $T$ but using the red triangles instead. Note that all vertices of $T$ are present in both $G$ and $H$, and we call these \emph{white vertices}. The vertices in $G$ corresponding to blue triangles we call \emph{blue vertices} and those in $H$ corresponding to red triangles we call \emph{red vertices}. The resulting graphs $G$ and $H$ are both planar, bipartite graphs. However, $G,H$ are not necessarily connected, and their faces may be non-simply connected.

If we know only the combinatorics of $G$, then the blue triangles of $T$ are determined and it only remains to triangulate the non-blue polygons remaining in $T$. As is well known, different such triangulations are related by edge flips of red triangles, and we shall see that these edge flips correspond to Miquel dynamics on the level of geometry in Section~\ref{sec:invariants}. Conversely, if we know only the combinatorics of $H$, then $T$ is determined up to blue edge flips, which correspond to the spider move on the dimer graph $G$, see also Section~\ref{sec:invariants}. Consequently, $G$ and $H$ determine each other up to dynamics, and the procedure is the same in both directions.

\subsection{Circle patterns}
In the following, we denote by $\CP^1 = \C \cup \{\infty\}$ the complex projective line, with a choice of affine coordinates given by $\C$.

\begin{definition}
	A \emph{circle pattern} $p$ on $H$ is an assignment of points $p(v) \in \CP^1$ to every vertex of $H$ and circles $p(f) \subset \CP^1$, such that $p(v)$ is on $p(f)$ whenever $v$ is incident to $f$ in $H$, and such that if two adjacent vertices $w,r$ are incident to the same two faces $f,f'$ then $p(w) \neq p(r)$ unless $p(f)$ is tangent to $p(f')$.
\end{definition}

Given a circle pattern $p: H \rightarrow \CP^1$, we call the map $t: F(H) \mapsto \C$ such that $t(f)$ is the center of $p(f)$ the corresponding \emph{center net} of $p$.
Note that the definition of centers depends on the choice of affine coordinates~$\C$. A center net is called a \emph{t-embedding} if it is an embedding of $H^*$ into $\C$. In particular, this implies that all circles are really circles in $\C$, rather than lines. As we will see, a special role is played by circle patterns such that $t$ is a \emph{convex t-embedding}, that is, if the image of every face is convex and no two overlap, and if no two adjacent vertices are mapped to the same point.
It is immediately clear that if $H$ has non-simply connected faces, then the associated t-realization is not a convex embedding.

\begin{figure}
	\centering
	\begin{tikzpicture}[scale=1.25,baseline=(current bounding box.center)]
		\coordinate[wvert] (w0) at (1,0.65);
		\coordinate[wvert] (w4) at (1,-1.7);
		\coordinate[wvert] (w5) at (3,1.7);
		\coordinate[wvert] (w6) at (-1,1.7);
		\filldraw[fill opacity=0.2,fill=red,draw=black,-]
			(w4.center) to[out=0,in=300,looseness=1.1] (w5.center) to (w0.center) -- cycle
			(w5.center) to[out=120,in=60,looseness=1.1] (w6.center) to (w0.center) -- cycle
			(w6.center) to[out=240,in=180,looseness=1.1] (w4.center) to (w0.center) -- cycle
		;

		\coordinate[bvert, red] (b4) at ($(w4)!.5!(w5)$);
		\coordinate[bvert, red] (b5) at ($(w5)!.5!(w6)$);
		\coordinate[bvert, red] (b6) at ($(w6)!.5!(w4)$);
		\draw[red]
			(b4) edge (w4) edge (w5)
			(b5) edge (w5) edge (w6)
			(b6) edge (w6) edge (w4)
            (w0) edge (b4) edge (b5) edge (b6)
		;
		\coordinate[wvert] (w0) at (w0); \coordinate[wvert] (w4) at (w4); \coordinate[wvert] (w5) at (w5); \coordinate[wvert] (w6) at (w6);

		\coordinate[] (t04) at ($(w0)!.5!(w4)$);
		\coordinate[] (t05) at ($(w0)!.5!(w5)$);
		\coordinate[] (t06) at ($(w0)!.5!(w6)$);
		\coordinate[] (t4) at ($(w0)!1.8!(b4)$);
		\coordinate[] (t5) at ($(w0)!1.8!(b5)$);
		\coordinate[] (t6) at ($(w0)!1.8!(b6)$);
		\draw[black, densely dashed]
			(t04) -- (t05) -- (t06) -- (t04)
			(t04) -- (t4) -- (t05)
			(t05) -- (t5) -- (t06)
			(t06) -- (t6) -- (t04)
		;

        \draw
			pic [draw=black, fill=darkgreen, angle radius=0.4cm] {angle=t4--t04--t05}
			pic [draw=black, fill=darkgreen, angle radius=0.4cm] {angle=t5--t05--t06}
			pic [draw=black, fill=darkgreen, angle radius=0.4cm] {angle=t6--t06--t04}
            pic [draw=black, fill=darkgreen, angle radius=0.4cm] {angle=t04--t05--t4}
            pic [draw=black, fill=darkgreen, angle radius=0.4cm] {angle=t05--t06--t5}
            pic [draw=black, fill=darkgreen, angle radius=0.4cm] {angle=t06--t04--t6}
		;
 		\coordinate[bvert,label=below left:$v_1$] (t04) at (t04);
        \coordinate[bvert,label=right:$v_2$] (t05) at (t05);
        \coordinate[bvert,label=left:$v_3$] (t06) at (t06);
 		\coordinate[bvert,label=below right:$u_1$] (t4) at (t4);
        \coordinate[bvert,label=right:$u_2$] (t5) at (t5);
        \coordinate[bvert,label=below left:$u_3$] (t6) at (t6);
	\end{tikzpicture}
    \hspace{5mm}
	\begin{tikzpicture}[scale=1.25,baseline=(current bounding box.center)]
		\coordinate[wvert] (w1) at (210:1);
		\coordinate[wvert] (w2) at (330:1);
		\coordinate[wvert] (w3) at (90:1);
		\filldraw[fill opacity=0.2,fill=blue,draw=black,-]
			(w1.center) -- (w2.center) -- (w3.center) -- cycle
		;
		\coordinate[wvert] (w4) at (270:2.3);
		\coordinate[wvert] (w5) at (30:2.3);
		\coordinate[wvert] (w6) at (150:2.3);
		\filldraw[fill opacity=0.2,fill=red,draw=black,-]
			(w1.center) -- (w4.center) -- (w2.center) -- cycle
			(w2.center) -- (w5.center) -- (w3.center) -- cycle
			(w3.center) -- (w6.center) -- (w1.center) -- cycle
			(w4.center) to[out=0,in=300,looseness=1.1] (w5.center) to (w2.center) -- cycle
			(w5.center) to[out=120,in=60,looseness=1.1] (w6.center) to (w3.center) -- cycle
			(w6.center) to[out=240,in=180,looseness=1.1] (w4.center) to (w1.center) -- cycle
		;

        \coordinate[label=below right:$u_{246}$] (t0) at (0,0);
		\coordinate[bvert, red] (b1) at ($(w4)!.5!(t0)$);
		\coordinate[bvert, red] (b2) at ($(w5)!.5!(t0)$);
		\coordinate[bvert, red] (b3) at ($(w6)!.5!(t0)$);
		\coordinate[bvert, red] (b4) at ($(t0)!1.65!(w2)$);
		\coordinate[bvert, red] (b5) at ($(t0)!1.65!(w3)$);
		\coordinate[bvert, red] (b6) at ($(t0)!1.65!(w1)$);
		\draw[red]
            (w4) edge (b4) edge (b6) edge (b1)
            (w5) edge (b5) edge (b4) edge (b2)
            (w6) edge (b6) edge (b5) edge (b3)
			(b1) edge (w1) edge (w2)
			(b2) edge (w2) edge (w3)
			(b3) edge (w3) edge (w1)
			(b4) edge (w2)
			(b5) edge (w3)
			(b6) edge (w1)
		;
		\coordinate[wvert] (w1) at (w1); \coordinate[wvert] (w2) at (w2); \coordinate[wvert] (w3) at (w3); \coordinate[wvert] (w4) at (w4); \coordinate[wvert] (w5) at (w5); \coordinate[wvert] (w6) at (w6);

		\coordinate[label=below left:$v_6$] (t14) at ($(w1)!.5!(w4)$);
		\coordinate[label=below right:$v_1$] (t24) at ($(w2)!.5!(w4)$);
		\coordinate[label=right:$v_2$] (t25) at ($(w2)!.5!(w5)$);
		\coordinate[label=above right:$v_3$] (t35) at ($(w3)!.5!(w5)$);
		\coordinate[label=above left:$v_4$] (t36) at ($(w3)!.5!(w6)$);
		\coordinate[label=left:$v_5$] (t16) at ($(w1)!.5!(w6)$);
		\coordinate[bvert, label=below right:$u_1$] (t24a) at ($(t0)!2.1!(w2)$);
		\coordinate[bvert, label=right:$u_3$] (t35a) at ($(t0)!2.1!(w3)$);
		\coordinate[bvert, label=below left:$u_5$] (t16a) at ($(t0)!2.1!(w1)$);
		  \coordinate[] (t14b) at ($(t14)!.5!30:(w1)$);
		  \coordinate[] (t25b) at ($(t25)!.5!30:(w2)$);
		  \coordinate[] (t36b) at ($(t36)!.5!30:(w3)$);
        \coordinate[] (t24b) at ($(t24)!.5!330:(w2)$);
        \coordinate[] (t35b) at ($(t35)!.5!330:(w3)$);
        \coordinate[] (t16b) at ($(t16)!.5!330:(w1)$);
        \draw[black, densely dashed]
			(t0) -- (t14) -- (t24) -- (t0)
			(t0) -- (t25) -- (t35) -- (t0)
			(t0) -- (t36) -- (t16) -- (t0)
			(t24) to[out=30, in=270] (t25)
			(t35) to[out=150, in=30] (t36)
			(t16) to[out=270, in=150] (t14)
			(t24) -- (t24a)
			(t14) -- (t16a)
			(t35) -- (t35a)
			(t25) -- (t24a)
			(t16) -- (t16a)
			(t36) -- (t35a)
		;

        \draw
			pic [draw=black, fill=darkgreen, angle radius=0.4cm] {angle= t24--t14--t0}
			pic [draw=black, fill=darkgreen, angle radius=0.4cm] {angle= t35--t25--t0}
			pic [draw=black, fill=darkgreen, angle radius=0.4cm] {angle= t16--t36--t0}
			pic [draw=black, fill=darkgreen, angle radius=0.4cm] {angle= t0--t24--t14}
            pic [draw=black, fill=darkgreen, angle radius=0.4cm] {angle= t0--t35--t25}
			  pic [draw=black, fill=darkgreen, angle radius=0.4cm] {angle= t0--t16--t36}
			pic [draw=black, fill=darkgreen, angle radius=0.4cm] {angle=t16a--t16--t16b}
			pic [draw=black, fill=darkgreen, angle radius=0.4cm] {angle=t24a--t24--t24b}
			pic [draw=black, fill=darkgreen, angle radius=0.4cm] {angle=t35a--t35--t35b}
			pic [draw=black, fill=darkgreen, angle radius=0.4cm] {angle=t36b--t36--t35a}
			pic [draw=black, fill=darkgreen, angle radius=0.4cm] {angle=t25b--t25--t24a}
			pic [draw=black, fill=darkgreen, angle radius=0.4cm] {angle=t14b--t14--t16a}
		;
        \coordinate[bvert] (t0) at (t0);
        \coordinate[bvert] (t14) at (t14); \coordinate[bvert] (t24) at (t24); \coordinate[bvert] (t25) at (t25); \coordinate[bvert] (t35) at (t35); \coordinate[bvert] (t36) at (t36); \coordinate[bvert] (t16) at (t16);

	\end{tikzpicture}
	\caption{Two examples of graphs $H$ (red) such that $t$ cannot be a convex t-embedding of $H^*$ (dashed). The angle sum $S$ (green) has to be both equal to and less than $3\pi$ in the left and $6\pi$ in the right picture, respectively, which provides a contradiction.}
	\label{fig:gconnected}
\end{figure}
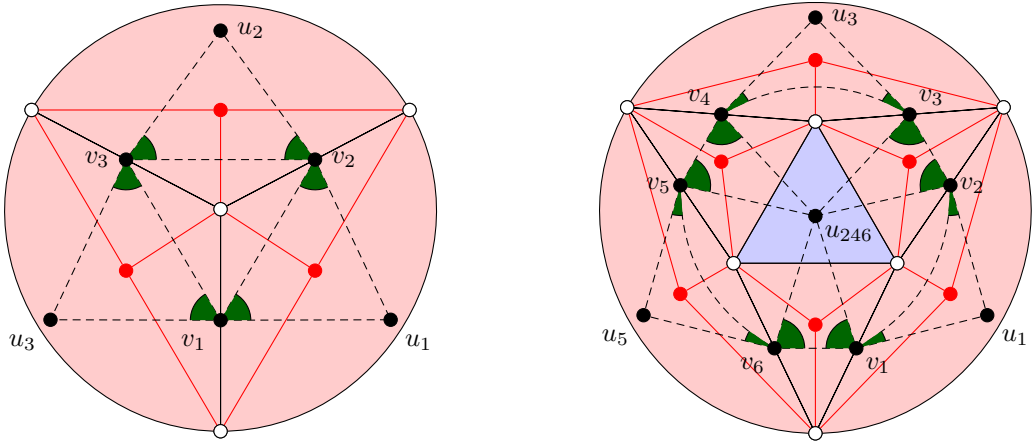

\begin{lemma}
	If $p$ is a circle pattern with convex t-embedding $t$, then $G$ has no non-simply connected faces.
\end{lemma}
\proof{
	Assume that $G$ has a non-simply connected face $f$, then the triangulation $T$ must have a sequence of consecutively adjacent red triangles which form a cycle $\gamma \in T^*$, see also Figure~\ref{fig:gconnected}.
    As a result, there is a sequence of degree four vertices $v_1,\dots,v_{m} \in H^*$ so that at each vertex $v_i$ there is an edge to $v_{i-1}$ and to $v_{i+1}$. Moreover, there is a disjoint sequence of (not necessarily distinct) vertices $u_1, \dots, u_m \in H^*$ so that for each $i$ there is an edge $(v_i,u_i)$ and $(v_{i+1},u_{i})$. Thus, we get a sequence of triangles $(v_i,u_i,v_{i+1})$ in $H^*$. Consider the sum of angles
    \begin{align}
        S = \sum_{i=1}^m \angle(u_i,v_i,v_{i+1}),
    \end{align}
    where each angle is understood to be in $]0,\pi/2[$. Since $t$ is a convex t-embedding, we know that $S = m \pi$. On the other hand, since the triangles $\Delta_i$ are also non-degenerate by assumption, we see that $S < m\pi$, which provides a contradiction. \qed
}

Consequently, whenever we are working with a circle pattern with convex t-embedding we may assume that both $G$ and $H$ have only simply connected faces. Moreover, if $G$ or $H$ are not connected, either the geometry or the dimer partition functions trivially factor, so that for simplicity we also assume in the following that $G$ and $H$ are connected.

\begin{remark}
    Note that if one allows the boundary circles to be straight lines, thus, for the ``outside'' vertices $u_i$ in Figure~\ref{fig:gconnected} to be mapped to $\infty$ by $t$, then the configurations may be realizable. For example, the configuration on the left of Figure~\ref{fig:gconnected} becomes a Miquel configuration.
\end{remark}

\subsection{Almost perfect matchings}

The \emph{dimer model} (see \cite{kenyondimerintro} for an introduction) is a model from statistical mechanics which samples perfect matchings on a graph according to a probability measure that is proportional to the product of the weights of the edges appearing in the respective perfect matching. For our purposes we consider \emph{almost perfect matchings} on $G$ (see \cite{postgrass}), which are subsets of the edges that cover every interior vertex of $G$ exactly once, but may not necessarily cover every boundary vertex. We denote by $\mathcal M_I$ the set of almost perfect matchings that cover the subset $I$ of the boundary vertices. Given a choice of edge weights $\lambda \in \C^\times$ the \emph{partition function} $Z_I$ is given by
\begin{align}
	Z_I = \sum_{M \in \mathcal M_I} \prod_{e \in M} \lambda(e).
\end{align}
If the edge weights are in $\R_{>0}$, then the probability measure on $\mathcal M_I$ is defined by
\begin{align}
	P_I(M) = \frac{1}{Z_I} \prod_{e \in M} \lambda(e).
\end{align}
A \emph{gauge transformation} of the edge weights is given by a function $\mu: I \rightarrow \C^\times$ with action
\begin{align}
    \lambda(v,v') &\mapsto \mu(v)\mu(v') \lambda(v,v'), &
    Z_I &\mapsto Z_I\prod_{v\in I} \mu(v).
\end{align}
Clearly, gauge transformations do not affect the probability measure. The \emph{face weights} $X: F(G) \rightarrow \R_{> 0}$ are defined as
\begin{align}
	X(f) = \frac{\lambda(w_1,b_1)\lambda(w_2,b_2)\cdots \lambda(w_n,b_n)}{\lambda(b_1,w_2)\lambda(b_2,w_3)\cdots \lambda(b_n,w_1)},
\end{align}
where the vertex labels are as in Figure~\ref{fig:facelabels}, that is, beginning at a white vertex and proceeding around the face $f$ in counterclockwise order.
It is well-known that two choices of edge weights $\lambda,\lambda'$ are gauge equivalent if and only if the corresponding face weights $X,X'$ agree on every face; and that every choice of face weights can be realized by an appropriate choice of edge weights.

A \emph{Kasteleyn orientation} $\kappa: E(G) \rightarrow \C^\times$ is an assignment of complex numbers to the edges of $G$, such that around every face $f$ of degree $2n$
\begin{align}
	\frac{\kappa(w_1,b_1)\kappa(w_2,b_2)\cdots \kappa(w_n,b_n)}{\kappa(b_1,w_2)\kappa(b_2,w_3)\cdots \kappa(b_n, w_1)} = (-1)^{n+1}.
\end{align}
Kasteleyn orientations were introduced because they allow to express the partition function $Z_I$ as the determinant of the weighted adjacency matrix of $G$ with weights given by $\kappa(e)\lambda(e)$. We do not need this determinant expression, nevertheless, the Kasteleyn orientation appears naturally in the geometry of circle patterns.
In the following we assume that $\kappa$ is an arbitrary fixed Kasteleyn orientation of $G$.

\subsection{Dimers and circle patterns}

Consider a blue vertex $b$ of $G$ with its three adjacent white vertices $w_1,w_2,w_3$ in counterclockwise order. Given a circle pattern $p$ on $H$, we define the \emph{edge coefficients} $\mu: E(G) \rightarrow \C$ by
\begin{align}
	\mu(w_1,b) &= p(w_2) - p(w_3), &
    \mu(w_2,b) &= p(w_3) - p(w_1), &
    \mu(w_3,b) &= p(w_1) - p(w_2).
\end{align}
Furthermore, we define the \emph{Möbius edge weights} of $G$ to equal $\lambda(w,b) = \kappa(w,b)\mu(w,b)$. In general, these edge weights are complex numbers, but we will see in a moment that they are gauge equivalent to real weights.

Notice that using the vectors $\hat p= (p,1) \in \C^2$ and the coefficients $\mu$ we obtain a so called \emph{vector-relation configuration} \cite{agprvrc}, which is an assignment of vectors $\hat p$ to the vertices of $G$ and coefficients $\mu$ to the edges of $G$ so that for each blue vertex the relation
\begin{align}
    \mu(w_1,b) \hat p(w_1) + \mu(w_2,b) \hat p(w_2) + \mu(w_3,b) \hat p(w_3) = 0
\end{align}
holds, which is evidently the case. We inherit the following useful lemma for vector-relation configurations.
\begin{lemma} \label{lem:facemr}
    The Möbius face weight $X(f)$ for a face of degree $2n$ is given by the multi-ratio
    \begin{align}
        X(f) &= (-1)^{n+1} \mr(p(w_1),p(v_1),p(w_2),p(v_2), \dots, p(w_n),p(v_n)) \\
        &= (-1)^{n+1}\frac{(p(w_1)-p(v_1))(p(w_2)-p(v_2)) \cdots (p(w_n)-p(v_n))}{(p(v_1) - p(w_2))(p(v_2)-p(w_3)) \cdots (p(v_n)-p(w_1))},
    \end{align}
    with the labeling as in Figure~\ref{fig:facelabels}.
\end{lemma}

As multi-ratios are invariant under (orientation preserving) Möbius transformations of $\CP^1$, an immediate consequence of Lemma~\ref{lem:facemr} is that the face weights are also invariant under Möbius transformations. As the next lemma shows, the face weights of a circle pattern are actually real and therefore also invariant under orientation reversing Möbius transformations.

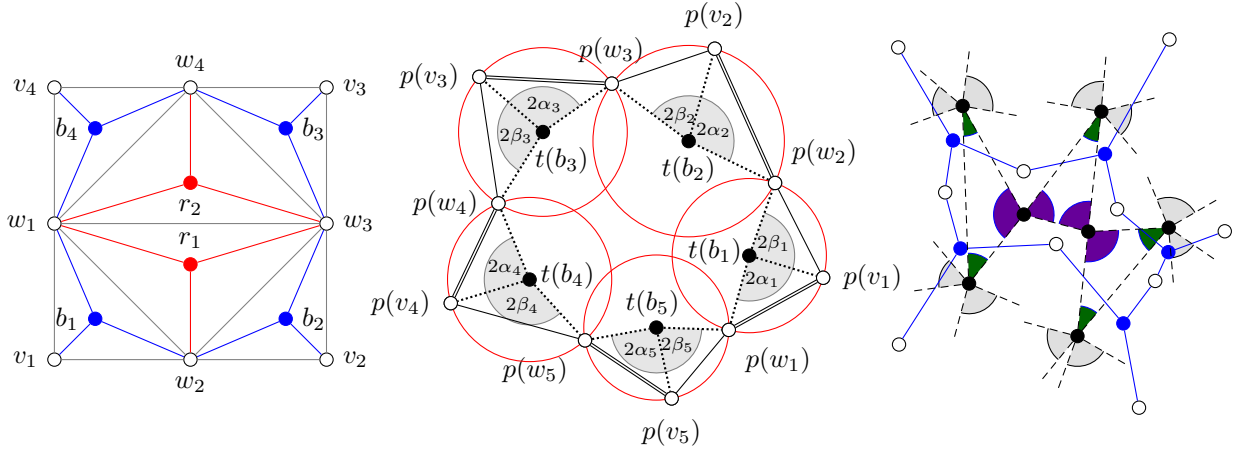
\begin{figure}
    \centering
    \begin{tikzpicture}[scale=1.8, baseline=(current bounding box.center)]
        \coordinate[wvert,label=left:{$v_1$}] (w0) at (0,0);
        \coordinate[wvert,label=below:{$w_2$}] (w1) at (1,0);
        \coordinate[wvert,label=right:{$v_2$}] (w2) at (2,0);
        \coordinate[wvert,label=right:{$w_3$}] (w3) at (2,1);
        \coordinate[wvert,label=right:{$v_3$}] (w4) at (2,2);
        \coordinate[wvert,label=above:{$w_4$}] (w5) at (1,2);
        \coordinate[wvert,label=left:{$v_4$}] (w6) at (0,2);
        \coordinate[wvert,label=left:{$w_1$}] (w7) at (0,1);
        \coordinate[bvert,label=left:{$b_1$}, blue] (b0) at (0.3, 0.3);
        \coordinate[bvert,label=right:{$b_2$}, blue] (b2) at (1.7, 0.3);
        \coordinate[bvert,label=right:{$b_3$}, blue] (b4) at (1.7, 1.7);
        \coordinate[bvert,label=left:{$b_4$}, blue] (b6) at (0.3, 1.7);
        \coordinate[bvert,label=above:{$r_1$}, red] (r1) at (1, 0.7);
        \coordinate[bvert,label=below:{$r_2$}, red] (r5) at (1, 1.3);
        \draw[blue]
            (b0) edge (w0) edge (w1) edge (w7)
            (b2) edge (w2) edge (w1) edge (w3)
            (b4) edge (w4) edge (w5) edge (w3)
            (b6) edge (w6) edge (w5) edge (w7)
        ;
        \draw[red]
            (r1) edge (w1) edge (w3) edge (w7)
            (r5) edge (w5) edge (w3) edge (w7)
        ;
        \draw[gray]
            (w0) -- (w1) -- (w2) -- (w3) -- (w4) -- (w5) -- (w6) -- (w7) -- (w0)
            (w1) -- (w3) -- (w5) -- (w7) -- (w1)    (w3) -- (w7)
        ;
    \end{tikzpicture}
    \hspace{-5mm}
    \begin{tikzpicture}[scale=0.6, baseline=(current bounding box.center)]
        \pgfdeclarelayer{background}
    	\pgfdeclarelayer{foreground}
    	\pgfsetlayers{background,main,foreground}
    	\def\radius{5}
        \begin{pgfonlayer}{foreground}
            \coordinate[wvert] (x1) at (1.9,-1.69);
            \coordinate[wvert] (x2) at (3.98,-0.54);
            \coordinate[wvert] (x3) at (2.92,1.55);
            \coordinate[wvert] (x4) at (1.59,4.51);
            \coordinate[wvert] (x5) at (-0.68,3.74);
            \coordinate[wvert] (x6) at (-3.6,3.89);
            \coordinate[wvert] (x7) at (-3.18,1.1);
            \coordinate[wvert] (x8) at (-4.23,-1.1);
            \coordinate[wvert] (x9) at (-1.27,-1.92);
            \coordinate[wvert] (x10) at (0.64,-3.2);
            \node at (x1) [below right = 0.5mm of x1] {$p(w_1)$};
            \node at (x2) [right = 0.5mm of x2] {$p(v_1)$};
            \node at (x3) [above right = 0.5mm of x3] {$p(w_2)$};
            \node at (x4) [above = 0.5mm of x4] {$p(v_2)$};
            \node at (x5) [above = 0.5mm of x5] {$p(w_3)$};
            \node at (x6) [left = 0.5mm of x6] {$p(v_3)$};
            \node at (x7) [left = 0.5mm of x7] {$p(w_4)$};
            \node at (x8) [left = 0.5mm of x8] {$p(v_4)$};
            \node at (x9) [below left = 0.5mm of x9] {$p(w_5)$};
            \node at (x10) [below = 0.5mm of x10] {$p(v_5)$};
            \draw[-]
                (x1)  (x2) -- (x3)  (x4) -- (x5)  (x6) -- (x7)  (x8) -- (x9)  (x10) -- (x1)
            ;
            \draw[-,double]
                (x1) -- (x2)  (x3) -- (x4)  (x5) -- (x6)  (x7) -- (x8)  (x9) -- (x10)  (x1)
            ;
            \tkzDefCircle[circum](x1,x2,x3)\tkzGetPoint{A}
            \tkzDefCircle[circum](x3,x4,x5)\tkzGetPoint{B}
            \tkzDefCircle[circum](x5,x6,x7)\tkzGetPoint{C}
            \tkzDefCircle[circum](x7,x8,x9)\tkzGetPoint{D}
            \tkzDefCircle[circum](x9,x10,x1)\tkzGetPoint{E}
            \coordinate[bvert] (z1) at (A);
            \coordinate[bvert] (z2) at (B);
            \coordinate[bvert] (z3) at (C);
            \coordinate[bvert] (z4) at (D);
            \coordinate[bvert] (z5) at (E);
            \draw[fill=none,red](z1) circle (1.7) node [black] {};
            \draw[fill=none,red](z2) circle (2.11) node [black] {};
            \draw[fill=none,red](z3) circle (1.86) node [black] {};
            \draw[fill=none,red](z4) circle (1.81) node [black] {};
            \draw[fill=none,red](z5) circle (1.6) node [black] {};
            \node at (z1) [label={[xshift=1.5mm,yshift=0.5mm]left:$t(b_1)$}] {};
            \node at (z2) [label={[xshift=0mm,yshift=0.5mm]below:$t(b_2)$}] {};
            \node at (z3) [label={[xshift=2.5mm,yshift=0.5mm]below:$t(b_3)$}] {};
            \node at (z4) [label={[xshift=-1mm,yshift=0.5mm]right:$t(b_4)$}] {};
            \node at (z5) [label={[xshift=-0.5mm,yshift=-1mm]above:$t(b_5)$}] {};
            \coordinate[wvert] (x1) at (1.9,-1.69);
            \coordinate[wvert] (x2) at (3.98,-0.54);
            \coordinate[wvert] (x3) at (2.92,1.55);
            \coordinate[wvert] (x4) at (1.59,4.51);
            \coordinate[wvert] (x5) at (-0.68,3.74);
            \coordinate[wvert] (x6) at (-3.6,3.89);
            \coordinate[wvert] (x7) at (-3.18,1.1);
            \coordinate[wvert] (x8) at (-4.23,-1.1);
            \coordinate[wvert] (x9) at (-1.27,-1.92);
            \coordinate[wvert] (x10) at (0.64,-3.2);
        \end{pgfonlayer}
        \draw[densely dotted, thick]
            (x1) -- (z1) -- (x2)
            (z1) -- (x3)
            (x3) -- (z2) -- (x4)
            (z2) -- (x5)
            (x5) -- (z3) -- (x6)
            (z3) -- (x7)
            (x7) -- (z4) -- (x8)
            (z4) -- (x9)
            (x9) -- (z5) -- (x1)
            (z5) -- (x10)
        ;
        \begin{pgfonlayer}{background}
            \scriptsize
            \draw
                pic ["$2\alpha_1$", draw=gray, fill=gray!20, angle radius=0.6cm] {angle= x1--z1--x2}
                pic ["$2\beta_1$", draw=gray, fill=gray!20, angle radius=0.6cm] {angle= x2--z1--x3}
                pic ["$2\alpha_2$", draw=gray, fill=gray!20, angle radius=0.6cm] {angle= x3--z2--x4}
                pic ["$2\beta_2$", draw=gray, fill=gray!20, angle radius=0.6cm] {angle= x4--z2--x5}
                pic ["$2\alpha_3$", draw=gray, fill=gray!20, angle radius=0.6cm] {angle= x5--z3--x6}
                pic ["$2\beta_3$", draw=gray, fill=gray!20, angle radius=0.6cm] {angle= x6--z3--x7}
                pic ["$2\alpha_4$", draw=gray, fill=gray!20, angle radius=0.6cm] {angle= x7--z4--x8}
                pic ["$2\beta_4$", draw=gray, fill=gray!20, angle radius=0.6cm] {angle= x8--z4--x9}
                pic ["$2\alpha_5$", draw=gray, fill=gray!20, angle radius=0.6cm] {angle= x9--z5--x10}
                pic ["$2\beta_{5}$", draw=gray, fill=gray!20, angle radius=0.6cm] {angle= x10--z5--x1}
            ;
        \end{pgfonlayer}
    \end{tikzpicture}
    \hspace{-5mm}
    \begin{tikzpicture}[scale=0.15, baseline=(current bounding box.center),rotate=10]
        \pgfdeclarelayer{background layer}
        \pgfsetlayers{background layer,main}
        \def\radius{5}
        \coordinate[wvert] (w0) at (-24.08,-6.8);
        \coordinate[wvert] (w1) at (-17.28,-3.55);
        \coordinate[wvert] (w2) at (-26.29,0.01);
        \coordinate[wvert] (w3) at (-19.15,13.96);
        \coordinate[wvert] (w4) at (-33.83,4.74);
        \coordinate[wvert] (w5) at (-42.83,17.42);
        \coordinate[wvert] (w6) at (-40.97,4.11);
        \coordinate[wvert] (w7) at (-47.31,-8.4);
        \coordinate[wvert] (w8) at (-32.12,-2.08);
        \coordinate[wvert] (w9) at (-27.43,-17.56);
        \coordinate (c) at (-30.58,1.25);
        \coordinate[bvert,blue] (b1) at (-22.5,-4.5);
        \coordinate[bvert,blue] (b3) at (-26.5,5);
        \coordinate[bvert,blue] (b5) at (-39.5,8.5);
        \coordinate[bvert,blue] (b7) at (-40.5,-1);
        \coordinate[bvert,blue] (b9) at (-27.5,-10);
        \tkzDefCircle[circum](w0,w1,w2)\tkzGetPoint{A}
        \tkzDefCircle[circum](w2,w3,w4)\tkzGetPoint{B}
        \tkzDefCircle[circum](w4,w5,w6)\tkzGetPoint{C}
        \tkzDefCircle[circum](w6,w7,w8)\tkzGetPoint{D}
        \tkzDefCircle[circum](w8,w9,w0)\tkzGetPoint{E}
        \tkzDefCircle[circum](w2,c,w8)\tkzGetPoint{F}
        \tkzDefCircle[circum](w4,c,w8)\tkzGetPoint{G}
        \coordinate[bvert] (t1) at (A);
        \coordinate[bvert] (t2) at (B);
        \coordinate[bvert] (t3) at (C);
        \coordinate[bvert] (t4) at (D);
        \coordinate[bvert] (t5) at (E);
        \coordinate[bvert] (t6) at (F);
        \coordinate[bvert] (t7) at (G);
        \draw[-,blue]
            (w0) -- (b1) -- (w2) -- (b3) -- (w4) -- (b5) -- (w6) -- (b7) -- (w8) -- (b9) -- (w0)
            (b1) -- (w1)
            (b3) -- (w3)
            (b5) -- (w5)
            (b7) -- (w7)
            (b9) -- (w9)
        ;
        \draw[-, densely dashed]
            (t5) -- (t6) -- (t1) -- (t5)
            (t2) -- (t6) -- (t7) -- (t2)
            (t3) -- (t7) -- (t4) -- (t3)
        ;
        \draw[-, densely dashed]
            (t1) -- ++(265:\radius) coordinate (t1a1)
            (t1) -- ++(265+50.48:\radius) coordinate (t1a2)
            (t1) -- ++(20:\radius) coordinate (t1b1)
            (t1) -- ++(20+82.41:\radius) coordinate (t1b2)
            (t2) -- ++(300:\radius) coordinate (t2a1)
            (t2) -- ++(300+64.24:\radius) coordinate (t2a2)
            (t2) -- ++(75:\radius) coordinate (t2b1)
            (t2) -- ++(75+84.98:\radius) coordinate (t2b2)
            (t3) -- ++(340:\radius) coordinate (t3a1)
            (t3) -- ++(340+57.86+35:\radius) coordinate (t3a2)
            (t3) -- ++(130:\radius) coordinate (t3b1)
            (t3) -- ++(130+59.72:\radius) coordinate (t3b2)
            (t4) -- ++(120:\radius) coordinate (t4a1)
            (t4) -- ++(120+57.55:\radius) coordinate (t4a2)
            (t4) -- ++(245:\radius) coordinate (t4b1)
            (t4) -- ++(245+81.93:\radius) coordinate (t4b2)
            (t5) -- ++(155:\radius) coordinate (t5a1)
            (t5) -- ++(155+23.12+42.71:\radius) coordinate (t5a2)
            (t5) -- ++(240:\radius) coordinate (t5b1)
            (t5) -- ++(240+80:\radius) coordinate (t5b2)
        ;
        \begin{pgfonlayer}{background layer}
            \draw
                pic [draw=blue, fill=darkviolet, angle radius=0.4cm] {angle= t5--t6--t1}
                pic [draw=blue, fill=darkgreen, angle radius=0.4cm] {angle= t6--t1--t5}
                pic [draw=blue, fill=darkgreen, angle radius=0.4cm] {angle= t1--t5--t6}
                pic [draw=blue, fill=darkviolet, angle radius=0.4cm] {angle= t2--t6--t7}
                pic [draw=blue, fill=darkviolet, angle radius=0.4cm] {angle= t6--t7--t2}
                pic [draw=blue, fill=darkgreen, angle radius=0.4cm] {angle= t7--t2--t6}
                pic [draw=blue, fill=darkviolet, angle radius=0.4cm] {angle= t3--t7--t4}
                pic [draw=blue, fill=darkgreen, angle radius=0.4cm] {angle= t7--t4--t3}
                pic [draw=blue, fill=darkgreen, angle radius=0.4cm] {angle= t4--t3--t7}
                pic [draw=black, fill=gray!25, angle radius=0.4cm] {angle= t1a1--t1--t1a2}
                pic [draw=black, fill=gray!25, angle radius=0.4cm] {angle= t1b1--t1--t1b2}
                pic [draw=black, fill=gray!25, angle radius=0.4cm] {angle= t2a1--t2--t2a2}
                pic [draw=black, fill=gray!25, angle radius=0.4cm] {angle= t2b1--t2--t2b2}
                pic [draw=black, fill=gray!25, angle radius=0.4cm] {angle= t3a1--t3--t3a2}
                pic [draw=black, fill=gray!25, angle radius=0.4cm] {angle= t3b1--t3--t3b2}
                pic [draw=black, fill=gray!25, angle radius=0.4cm] {angle= t4a1--t4--t4a2}
                pic [draw=black, fill=gray!25, angle radius=0.4cm] {angle= t4b1--t4--t4b2}
                pic [draw=black, fill=gray!25, angle radius=0.4cm] {angle= t5a1--t5--t5a2}
                pic [draw=black, fill=gray!25, angle radius=0.4cm] {angle= t5b1--t5--t5b2}
            ;
            \end{pgfonlayer}
    \end{tikzpicture}

    \caption{Left: standard labeling for faces of $G$ (blue) and corresponding part of $H$ (red). Center: angles (gray) that determine $X(f)$. Right: angles in the t-embedding (dashed) that correspond to $\alpha_k,\beta_k$ (gray), and $\gamma_k$ (green) and supplementary angles (violet) used for the proof of the positivity of $X(f)$.}
    \label{fig:facelabels}
\end{figure}

\begin{lemma} \label{lem:realfaceweights}
	The face weights defined by a circle pattern are real.
\end{lemma}
\proof{
    Consider the face weight $X(f)$ of a face $f$ of $G$ of degree $2n$.
    It is well known that the cross-ratio of four points is real if and only if the four points are on a circle, which already proves the case $n=2$, since the four vertices involved are always in a common face of $H$. For the case $n>2$, consider the restriction of $G \cup H$ to the face $f$. The resulting graph is a quadrangulation. The multi-ratio that expresses $X(f)$ can be decomposed into a product of cross-ratios, one cross-ratio for each face of the quadrangulation. For example, for the situation in Figure~\ref{fig:facelabels} (left) we obtain
    \begin{align}
        X(f) = -\cro(w_1,v_1,w_2,r_1)\cro(w_2,v_2,w_3,r_1)\cro(w_3,v_3,w_4,r_2)\cro(w_4,v_4,w_1,r_2)\cro(w_1,r_1,w_3,r_2).
    \end{align}
    As each such cross-ratio is real, so is the face weight.\qed
}

\begin{remark}
    There is a converse of Lemma~\ref{lem:realfaceweights}: every vector-relation configuration with graph $G$ and real face weights in $\C^2$ defines a circle pattern in $\CP^1 = (\C^2 \setminus \{0\}) / \C^\times$ with the combinatorics of $H$. As above, one identifies $\C^2$ with the homogeneous coordinate space of $\CP^1$ which contains $\C$. The proof of the claim proceeds analogously to the proof of Lemma~\ref{lem:realfaceweights}: one consecutively chooses the images of the red vertices as the other intersection point of two already determined circles. For the last point there is a compatibility condition, which is exactly that the face weight is real.
\end{remark}

Of course, to define a probability measure the face weights need to be real and \emph{positive}. As we will see in the following, this is true for circle patterns with convex t-embedding.

\begin{lemma} \label{lem:mrsin}
    The Möbius face weight of a face $f$ of degree $2n$ is
    \begin{align} 
        X(f) = (-1)^{n+1} \exp\left(- i\sum_{k=1}^{n} \alpha_k + \beta_k \right) \prod_{k=1}^{n} \frac{\sin (\alpha_{k})}{\sin (\beta_{k})},
    \end{align}
    with the angle labels as in Figure~\ref{fig:facelabels} (center).
\end{lemma}
\proof{
    Follows because
    \begin{align}
        \frac{p(w_{k})-p(v_k)}{p(v_k)-p(w_{k+1})} = \exp\left(-i(\alpha_k+\beta_k)\right)\frac{\sin(\alpha_k)}{\sin(\beta_k)},
    \end{align}
    by invoking the law of sines for the absolute value and by the inscribed angle theorem for the argument.\qed
}

\begin{theorem} \label{thm:positivity}
	The Möius face weights of a circle pattern with convex t-embedding are real positive.
\end{theorem}
\proof{
    See Figure~\ref{fig:facelabels} (center and right) for an illustration of the proof.
    Consider a face $f$ of degree $2n$. Note that every blue vertex $b_k \in G$ corresponds to a face of $H$, and thus to a circle $c_k$ with center $t_k$. The point $p(w_k)$ is obtained from $p(w_{k+1})$ by two reflections about two consecutive edges of the t-embedding, or equivalently by a rotation around $t(b_k)$ with some angle $\gamma_k$. Similarly, $p(v_k)$ is obtained from $p(w_{k})$ by an even number of reflections about consecutive edges of the t-embedding, the angle of the corresponding rotation being $\alpha_k$. Analogously, $p(w_{k+1})$ is obtained from $p(v_{k})$ by a rotation of angle $\beta_k$. Since the t-embedding is convex, all three angles $\alpha_k, \beta_k$ and $\gamma_k$ are in the open interval $]0, \pi[$ and the sum of the three angles is $\pi$. By using triangle and vertex sum equalities for convex t-embeddings, we obtain that the sum of all $\gamma_k$ is
    \begin{align}
        (n-2)\pi - (n-3)\pi = \pi,
    \end{align}
    thus, the sum of all $\alpha_k,\beta_k$ is
    \begin{equation}
        \sum_{k=1}^n \alpha_k+\beta_k = (n-1)\pi.
    \end{equation}
    Together with Lemma~\ref{lem:mrsin} this proves the claim. \qed
}

\begin{remark}
    An additional consequence of Lemma~\ref{lem:mrsin} is that the Möbius face weights are already determined by the t-embedding, as the angles $\alpha_k$, $\beta_k$ can be read off the t-embedding.
\end{remark}

\section{Dynamics and invariants of motion}
\label{sec:invariants}

\begin{figure}
    \centering
    \begin{tikzpicture}[scale=2.2,baseline=(current bounding box.center)]
        \coordinate[wvert,label=below:{$w_1$}] (v00) at (0,0);
        \coordinate[wvert,label=below:{$w_2$}] (v10) at (1,0);
        \coordinate[wvert,label=above:{$w_4$}] (v01) at (0,1);
        \coordinate[wvert,label=above:{$w_3$}] (v11) at (1,1);
        \draw[-]
            (v00) -- (v10) -- (v11) -- (v01) -- (v00) -- (v11)
        ;
        \filldraw[opacity=0.25,fill=blue]
            (v00.center) -- (v10.center) -- (v11.center) -- cycle
            (v00.center) -- (v01.center) -- (v11.center) -- cycle
        ;
        \coordinate[wvert] (v00) at (0,0); \coordinate[wvert] (v10) at (1,0);
        \coordinate[wvert] (v01) at (0,1); \coordinate[wvert] (v11) at (1,1);

        \coordinate[bvert, blue,label=below:{$b_1$}] (b00) at (0.7,0.3);
        \coordinate[bvert, blue,label=above:{$b_2$}] (r00) at (0.3,0.7);
        \draw[blue]
            (b00) edge (v00) edge (v10) edge (v11)
            (r00) edge (v00) edge (v01) edge (v11)
        ;
    \end{tikzpicture}
    \hspace{-3.5mm}$\leftrightarrow$\hspace{-3.5mm}
    \begin{tikzpicture}[scale=2.2, rotate=90,baseline=(current bounding box.center)]
        \coordinate[wvert,label=below:{$w_2$}] (v00) at (0,0);
        \coordinate[wvert,label=above:{$w_3$}] (v10) at (1,0);
        \coordinate[wvert,label=below:{$w_1$}] (v01) at (0,1);
        \coordinate[wvert,label=above:{$w_4$}] (v11) at (1,1);
        \draw[-]
            (v00) -- (v10) -- (v11) -- (v01) -- (v00) -- (v11)
        ;
        \filldraw[opacity=0.25,fill=blue]
            (v00.center) -- (v10.center) -- (v11.center) -- cycle
            (v00.center) -- (v01.center) -- (v11.center) -- cycle
        ;
        \coordinate[wvert] (v00) at (0,0); \coordinate[wvert] (v10) at (1,0);
        \coordinate[wvert] (v01) at (0,1); \coordinate[wvert] (v11) at (1,1);

        \coordinate[bvert, blue,label=above:{$b_3$}] (b00) at (0.7,0.3);
        \coordinate[bvert, blue,label=below:{$b_4$}] (r00) at (0.3,0.7);
        \draw[blue]
            (b00) edge (v00) edge (v10) edge (v11)
            (r00) edge (v00) edge (v01) edge (v11)
        ;
    \end{tikzpicture}
    \hspace{-1mm}\vrule\hspace{-1.5mm}
    \begin{tikzpicture}[scale=2.2,baseline=(current bounding box.center)]
        \coordinate[wvert,label=below:{$w_1$}] (v00) at (0,0);
        \coordinate[wvert,label=below:{$w_2$}] (v10) at (1,0);
        \coordinate[wvert,label=above:{$w_4$}] (v01) at (0,1);
        \coordinate[wvert,label=above:{$w_3$}] (v11) at (1,1);
        \draw[-]
            (v00) -- (v10) -- (v11) -- (v01) -- (v00) -- (v11)
        ;
        \filldraw[opacity=0.25,fill=red]
            (v00.center) -- (v10.center) -- (v11.center) -- cycle
            (v00.center) -- (v01.center) -- (v11.center) -- cycle
        ;
        \coordinate[wvert] (v00) at (0,0); \coordinate[wvert] (v10) at (1,0);
        \coordinate[wvert] (v01) at (0,1); \coordinate[wvert] (v11) at (1,1);

        \coordinate[bvert, red,label=below:{$r_1$}] (b00) at (0.7,0.3);
        \coordinate[bvert, red,label=above:{$r_2$}] (r00) at (0.3,0.7);
        \draw[red]
            (b00) edge (v00) edge (v10) edge (v11)
            (r00) edge (v00) edge (v01) edge (v11)
        ;
    \end{tikzpicture}
    \hspace{-3.5mm}$\leftrightarrow$\hspace{-3.5mm}
    \begin{tikzpicture}[scale=2.2,rotate=90,baseline=(current bounding box.center)]
        \coordinate[wvert,label=below:{$w_2$}] (v00) at (0,0);
        \coordinate[wvert,label=above:{$w_3$}] (v10) at (1,0);
        \coordinate[wvert,label=below:{$w_1$}] (v01) at (0,1);
        \coordinate[wvert,label=above:{$w_4$}] (v11) at (1,1);
        \draw[-]
            (v00) -- (v10) -- (v11) -- (v01) -- (v00) -- (v11)
        ;
        \filldraw[opacity=0.25,fill=red]
            (v00.center) -- (v10.center) -- (v11.center) -- cycle
            (v00.center) -- (v01.center) -- (v11.center) -- cycle
        ;
        \coordinate[wvert] (v00) at (0,0); \coordinate[wvert] (v10) at (1,0);
        \coordinate[wvert] (v01) at (0,1); \coordinate[wvert] (v11) at (1,1);

        \coordinate[bvert, red,label=above:{$r_3$}] (b00) at (0.7,0.3);
        \coordinate[bvert, red,label=below:{$r_4$}] (r00) at (0.3,0.7);
        \draw[red]
            (b00) edge (v00) edge (v10) edge (v11)
            (r00) edge (v00) edge (v01) edge (v11)
        ;
    \end{tikzpicture}
    \hspace{-1mm}\vrule\hspace{-1.5mm}
    \begin{tikzpicture}[scale=1.1,baseline=(current bounding box.center)]
        \coordinate[wvert,label=below:{$w_1$}] (v00) at (0,0);
        \coordinate[wvert,label=below:{$w_2$}] (v20) at (2,0);
        \coordinate[wvert] (v11) at (1,1);
        \coordinate[wvert,label=above:{$w_4$}] (v02) at (0,2);
        \coordinate[wvert,label=above:{$w_3$}] (v22) at (2,2);
        \draw[-]
            (v00) -- (v20) -- (v22) -- (v02) -- (v00)
            (v11) edge (v00) edge (v20) edge (v22) edge (v02)
        ;
        \filldraw[opacity=0.25,fill=blue]
            (v00.center) -- (v20.center) -- (v11.center) -- cycle
            (v22.center) -- (v02.center) -- (v11.center) -- cycle
        ;
        \filldraw[opacity=0.25,fill=red]
            (v02.center) -- (v00.center) -- (v11.center) -- cycle
            (v20.center) -- (v22.center) -- (v11.center) -- cycle
        ;
        \coordinate[wvert] (v00) at (0,0);
        \coordinate[wvert] (v20) at (2,0);
        \coordinate[wvert,label=below:{$w_5$}] (v11) at (1,1);
        \coordinate[wvert] (v02) at (0,2);
        \coordinate[wvert] (v22) at (2,2);

        \coordinate[bvert, blue,label={[xshift=0mm,yshift=1mm]below:$b_1$}] (b00) at (1,0.4);
        \coordinate[bvert, blue,label={[xshift=0mm,yshift=-1mm]above:$b_2$}] (b11) at (1,1.6);
        \draw[blue]
            (b00) edge (v00) edge (v20) edge (v11)
            (b11) edge (v11) edge (v02) edge (v22)
        ;

        \coordinate[bvert, red,label={[xshift=1mm,yshift=0mm]left:$r_1$}] (r00) at (0.4,1);
        \coordinate[bvert, red,label={[xshift=-1mm,yshift=0mm]right:$r_2$}] (r11) at (1.6,1);
        \draw[red]
            (r00) edge (v00) edge (v02) edge (v11)
            (r11) edge (v11) edge (v20) edge (v22)
        ;
        \coordinate[wvert,label=below:{$w_5$}] (v11) at (1,1);
    \end{tikzpicture}
    \hspace{-3.5mm}$\leftrightarrow$\hspace{-3.5mm}
    \begin{tikzpicture}[scale=1.1,rotate=90,baseline=(current bounding box.center)]
        \coordinate[wvert,label=below:{$w_2$}] (v00) at (0,0);
        \coordinate[wvert,label=above:{$w_3$}] (v20) at (2,0);
        \coordinate[wvert,label=below:{$w_6$}] (v11) at (1,1);
        \coordinate[wvert,label=below:{$w_1$}] (v02) at (0,2);
        \coordinate[wvert,label=above:{$w_4$}] (v22) at (2,2);
        \draw[-]
            (v00) -- (v20) -- (v22) -- (v02) -- (v00)
            (v11) edge (v00) edge (v20) edge (v22) edge (v02)
        ;
        \filldraw[opacity=0.25,fill=blue]
            (v00.center) -- (v20.center) -- (v11.center) -- cycle
            (v22.center) -- (v02.center) -- (v11.center) -- cycle
        ;
        \filldraw[opacity=0.25,fill=red]
            (v02.center) -- (v00.center) -- (v11.center) -- cycle
            (v20.center) -- (v22.center) -- (v11.center) -- cycle
        ;
        \coordinate[bvert, blue,label={[xshift=-1mm,yshift=0mm]right:$b_3$}] (b00) at (1,0.4);
        \coordinate[bvert, blue,label={[xshift=1mm,yshift=0mm]left:$b_4$}] (b11) at (1,1.6);
        \draw[blue]
            (b00) edge (v00) edge (v20) edge (v11)
            (b11) edge (v11) edge (v02) edge (v22)
        ;

        \coordinate[bvert, red,label={[xshift=0mm,yshift=1mm]below:$r_3$}] (r00) at (0.4,1);
        \coordinate[bvert, red,label={[xshift=0mm,yshift=-1mm]above:$r_4$}] (r11) at (1.6,1);
        \draw[red]
            (r00) edge (v00) edge (v02) edge (v11)
            (r11) edge (v11) edge (v20) edge (v22)
        ;
        \coordinate[wvert] (v00) at (0,0);
        \coordinate[wvert] (v20) at (2,0);
        \coordinate[wvert,label=below:{$w_6$}] (v11) at (1,1);
        \coordinate[wvert] (v02) at (0,2);
        \coordinate[wvert] (v22) at (2,2);
    \end{tikzpicture}

    \begin{tikzpicture}[line cap=round,line join=round,>=triangle 45,x=1.0cm,y=1.0cm,scale=0.55,baseline=(current bounding box.center)]
        \clip(2.32,-5.64) rectangle (9.36,1.48);
        \draw [thick, red] (5.774398596491228,-2.0514212865497083) circle (3.2922875115844032cm);
        \node[wvert, label=below:$w_1$] at (7.52,0.74) {};
        \node[wvert, label=right:$w_2$] at (2.8,-0.64) {};
        \node[wvert, label=above:$w_3$] at (5.38,-5.32) {};
        \node[wvert, label=left:$w_4$] at (8.83583274633734,-3.262527394655284) {};
    \end{tikzpicture}
    \hspace{8mm}
    \begin{tikzpicture}[line cap=round,line join=round,>=triangle 45,x=1.0cm,y=1.0cm,scale=0.3,baseline=(current bounding box.center)]
        \clip(0.7514928602159356,-8.56418237798597) rectangle (17.93895440550351,5.3792880809393875);
        \draw [thick,red]
            (5.30446280598748,1.0255105702953464) circle (3.917646438461537cm)
            (6.528073478913583,-5.206367042979455) circle (3.1114853037953076cm)
            (8.02,-2.48) circle (3.500021679963008cm)
            (11.450972232779066,2.1637530567382326) circle (2.6132479051252204cm)
            (13.364254064584326,-3.70496323392238) circle (4.198756223154831cm)
            (10.251912282926162,-0.40839873453853087) circle (3.3323445896792347cm)
        ;
        \node[wvert, label=left:$w_1$] at (4.536149127638532,-2.816057821449394) {};
        \node[wvert, label=right:$w_2$] at (9.615143301939945,-5.59539236798308) {};
        \node[wvert, label=left:$w_4$] at (8.87931516676848,2.628128495881128) {};
        \node[wvert, label=below right:$w_3$] at (10.812777623117984,-0.37036853739272857) {};
        \node[bvert, red, label=right:$r_2$] at (9.216137742191703,0.8092865885957101) {};
        \node[bvert, red, label=above right:$r_3$] at (13.460111778715287,0.49269862955864285) {};
        \node[bvert, red, label={[xshift=0.5mm]above:$r_4$}] at (7.468296964312057,-2.240339296440437) {};
        \node[bvert, red, label=left:$r_1$] at (9.168016951511206,-3.5595400158684156) {};
    \end{tikzpicture}
    \hspace{3mm}
    \begin{tikzpicture}[line cap=round,line join=round,>=triangle 45,x=1.0cm,y=1.0cm,scale=0.3,baseline=(current bounding box.center)]
        \clip(-2.4003130683302087,-6.968988788264981) rectangle (13.208141928863986,7.8987259717742475);
        \draw [thick,red]
            (2.488572496892653,4.38814634782904) circle (3.126765796491817cm)
            (7.16,3.64) circle (3.7566474415361366cm)
            (2.16,-2.2) circle (4.249470555257443cm)
            (7.872166720501163,-1.6356590807491307) circle (5.0871312341506405cm)
        ;
        \node[bvert, red, label=right:$r_{1234}$] at (3.98,1.64) {};
        \node[wvert, label=left:$w_4$] at (4.763643853326435,6.533074008438053)  {};
        \node[wvert, label=left:$w_1$] at (0.731107938490542,1.8020329179749701) {};
        \node[wvert, label=right:$w_2$] at (4.696235615136668,-5.609620052808278) {};
        \node[wvert, label=right:$w_3$] at (10.756479225202897,2.554763996780423) {};
        \node[wvert, label={[xshift=-1mm]268:$w_5$}] at (5.673775479819172,0.18984976042879165) {};
        \node[wvert, label={[xshift=-1mm]92:$w_6$}] at (5.052581864428122,2.598587656614506) {};
    \end{tikzpicture}
    \caption{From left to right: blue edge flip, red edge flip, and color flip.}
    \label{fig:moves}
\end{figure}
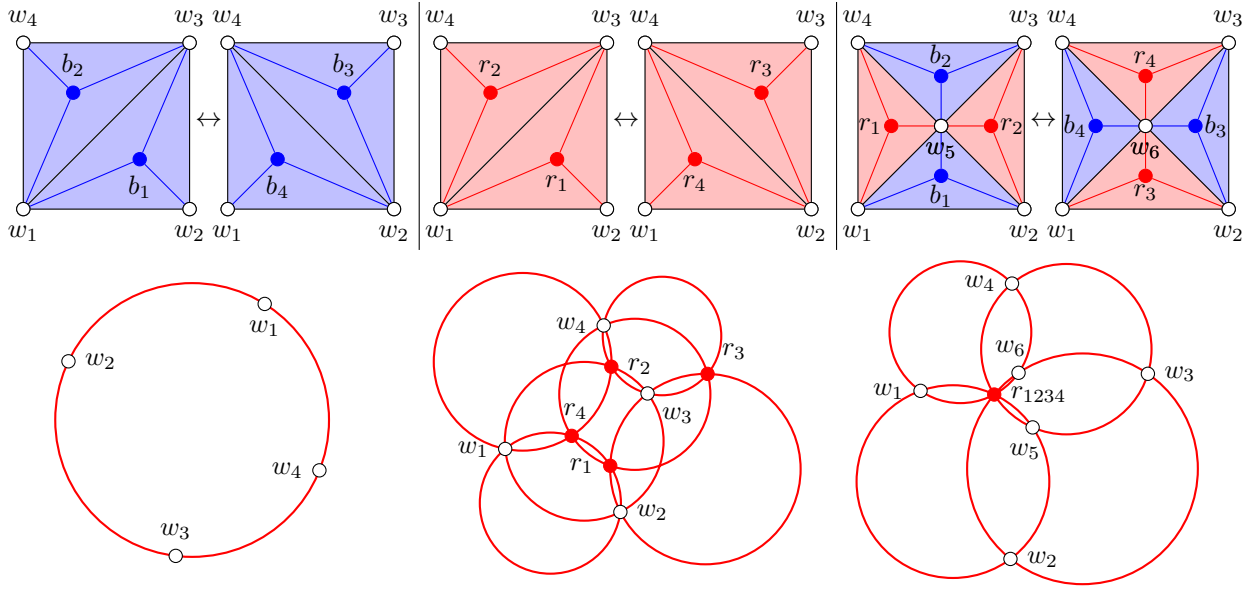

In this section, we introduce certain discrete time dynamics, discuss their geometric interpretation and show that dimer partition functions provide invariants of motion.
The dynamics are generated by local changes of combinatorics and geometry called \emph{flips}.
We allow three types of flips, which are local modifications of the bicolored triangulation $T$, see also Figure~\ref{fig:moves}:
\begin{enumerate}
	\item a \emph{blue edge flip:} an edge flip in two adjacent blue triangles,
	\item a \emph{red edge flip:} an edge flip in two adjacent red triangles,
	\item a \emph{color flip:} a color flip at a degree 4 vertex incident to two blue and two red triangles in alternating fashion.
\end{enumerate}

Let us first discuss the geometry of these flips, we refer to Figure~\ref{fig:moves} for the labeling. A blue edge flip does not change the combinatorics of $H$, therefore it is not visible in the geometry. A red edge flip does change the combinatorics of $H$, in particular, $p(r_1)$ and $p(r_2)$ which were the intersection points of two incident circles each are replaced by the two new points $p(r_3)$ and $p(r_4)$ which are the intersection points of the other two pairings of incident circles. A new circle corresponding to the central face exists due to Miquel's (six circle) theorem. For the color flip note that the four circles before and after the flip are unchanged, consequently, the four points $p(r_1)$, $p(r_2)$, $p(r_3)$ and $p(r_4)$ coincide. The point $p(w_5)$ is replaced by $p(w_6)$ which is the other intersection point besides $p(r_i)$ of two of the circles.
The point $p(w_6)$ is also algebraically determined by the following lemma.
\begin{lemma}[\cite{ksclifford,amiquelmobius}]
	The points of the six white vertices in the color flip satisfy the \emph{dSKP equation}
    \begin{equation}
        \mr(p(w_1),p(w_2),p(w_5),p(w_3),p(w_4),p(w_6)) = -1.\qedhere
    \end{equation}
\end{lemma}

\begin{remark}
    The dSKP equation \cite{ncwqint,dndskp,ksclifford} is known to be a discrete integrable equation in the sense of multi-dimensional consistency \cite{ksclifford,absoctahedron}. In particular, this implies that if we consider a sequence of circle patterns related by consecutive application of flips, so that the initial and the final circle pattern have the same combinatorics, then the geometry also coincides. Curiously, the dSKP equation also governs the points of the t-embedding under the red edge flip \cite{amiquel,klrrdimers}. The two occurrences of the dSKP equation are related on a higher dimensional $A_n$ lattice, see \cite{amiquelmobius}.
\end{remark}
Another important feature of the flips is that if a circle pattern had a convex t-embedding before a flip, it also has a convex t-embedding after the flip. For the blue edge flip and the color flip this is trivial, because the circles do not change. For the red edge flip this was shown in \cite{klrrdimers}.

Now let us consider the effect of flips on the graph $G$ carrying the dimer model. Clearly, a red edge flip does neither change the combinatorics of $G$ nor the edge weights, therefore both the face weights and the dimer partition function are unaffected. A blue edge flip corresponds to a well-known move for the dimer model, known as \emph{spider move} (also called urban renewal or square move). As mentioned before, our choice of edge weights makes $G$ into a vector-relation configuration \cite{agprvrc}. For vector-relation configurations it was already shown that the spider move acts on the face weights in the usual manner, which preserves correlations away from the spider move. Similarly, the color flip corresponds to a composition of a contraction and a split of the central vertex, also called \emph{resplit}. It was shown in \cite{agrtcd} that if the points involved in the resplit satisfy the dSKP equation, the face weights are unchanged, and therefore the resplit also preserves correlations. The general phenomenon is known as \emph{Z-invariance}, since it leaves the partition function $Z_I$ invariant up to gauge. In our case, we have a canonical gauge for the edge weights, which allows us to make the action on the partition function $Z$ concrete:
\begin{theorem}
    The action of a blue edge flip on $Z_I$ is
    \begin{align}
        Z_I \mapsto A \cdot Z_I, \qquad A =  \frac{p(w_1)-p(w_3)}{p(w_2)-p(w_4)},
    \end{align}
    and the action of a color flip on $Z_I$ is
    \begin{equation}
        Z_I \mapsto B \cdot Z_I, \qquad B = \frac{p(w_1)p(w_3)-p(w_2)p(w_4)-p(w_5)(p(w_1)-p(w_2)+p(w_3)-p(w_4))}{(p(w_1)-p(w_4))(p(w_2)-p(w_3))}.\qedhere
    \end{equation}
\end{theorem}
\proof{
    Direct calculation.\qed
}

As an immediate consequence, we get the following invariants of motion.

\begin{corollary}
    Given two subsets $I,J$ of the boundary vertices the ratio of partition functions $Z_I/Z_J$ constitutes an invariant of motion with respect to flips.
\end{corollary}

Let $M$ be a Möbius transformation and consider the circle pattern $M(p)$. By considering $p$ as a vector-relation configuration $\hat p$, we can represent $M$ as an element $\hat M \in \mathrm{PGL}(2,\C)$. Clearly, $\hat M \hat p$ is again a vector-relation configuration using the same edge coefficients and with the same face weights as $\hat p$. Hence, the canonical edge weights of $M(p)$ are related by a gauge transformation to those of $p$. As a consequence, we obtain the following corollary.

\begin{corollary}\label{cor:projinvariants}
    Given two families of subsets $(I_k)_{k=1}^m$, $(J_k)_{k=1}^m$ of the boundary vertices, such that the total multiplicity of each boundary vertex in both families is equal, the alternating ratio
    \begin{align}
        \frac{Z_{I_1}\cdot Z_{I_2} \cdots Z_{I_m}}{Z_{J_1}\cdot Z_{J_2} \cdots Z_{J_m}},
    \end{align}
    is an invariant of motion with respect to flips and a Möbius invariant of the circle pattern.
\end{corollary}

\begin{remark}
    In the case that $T$ consists only of two blue triangles with vertices $w_1,w_2,w_3,w_4$, it is not difficult to verify that
    \begin{align}
        \frac{Z_{w_1w_2}Z_{w_3w_4}}{Z_{w_2w_3}Z_{w_4w_1}} = -\cro(p(w_1), p(w_2), p(w_3), p(w_4)).
    \end{align}
    Permuting the vertices on the left corresponds to permuting the vertices on the right.
    The left-hand side is one of the invariants of motion discussed in Corollary~\ref{cor:projinvariants}, while the right-hand side coincides with $X(f)$ if $w_1,w_2,w_3$ and $w_3,w_4,w_1$ each share a blue vertex. It may be that similar formulas hold for larger minimal graphs, possibly as a manifestation of the \emph{twist} \cite{mstwist}.
\end{remark}

\section{$\mathrm{PO}(3,3)$ symmetry}
\label{sec:symmetry}

As discussed before, the (Möbius) face weights $X(f)$ we associate to a circle pattern are invariant under Möbius transformations (under a $\mathrm{PO}(3,1)$ action). On the other hand, for the (Euclidean) face weights associated to t-realizations in earlier work \cite{amiquel,klrrdimers}, Chelkak, Laslier and Russkikh \cite{clrmaximal,clrdimers} showed \emph{Lorentz invariance}, that is, invariance under a $\mathrm{O}(2,2)$ action. In this section, we show that the face weights $X(f)$ also exhibit Lorentz invariance, and in fact invariance under a more general $\mathrm{PO}(3,3)$ action. We show this by using the \emph{Lorentz lift} to $\R^{2,2}$ introduced in \cite{clrmaximal,clrdimers} and its extension to the circle pattern as discussed in \cite{admpsiso}, as well as subsequent inverse stereographic projection to a quadric of signature $\mathtt{(3,3)}$ in $\RP^5$. We describe the Lorentz lift in a slightly different fashion to facilitate the proof of the $\mathrm{PO}(3,3)$ invariance of face weights.

Given a circle pattern $p: V(H) \mapsto \C \simeq \R^2$, we use the \emph{origami connection} $\Gamma: E \mapsto \mathrm{Iso}(\R^2)$, where $\Gamma(v,v')$ is the reflection about the perpendicular bisector of $p(v)$ and $p(v')$, which is the line through two points of the corresponding t-realization. This connection is set up so that $p(v') = \Gamma(v,v') p(v)$ whenever $v,v'$ are adjacent and is well known to be flat. To each vertex $v$ of $H$ we also associate an (oriented) orthonormal frame $\Omega(v) = (x,y)(v) \in \S^1 \times \S^1$, such that $\Omega(v') = \Gamma(v,v') \Omega(v)$ whenever $v,v'$ are adjacent. To each face $f\in F(H)$ we associate two radii $r,\rho$ as follows: choose an incident vertex $v \in V(H)$ and let $r$ (resp.~$\rho$) be the signed radius of the circle centered at $t(f)$ in oriented contact to the line through $p(v)$ in direction $x(v)$ (resp.~$y(v)$) with normal $y(v)$ (resp.~$x(v)$). Due to the nature of the origami connection, the choice of $v$ does not matter.

The Lorentz lift $\theta: F(H) \mapsto \R^{2,2}$ is defined by $\theta(f) = (t, r, \rho)(f)$. For $v$ incident to $f$ the point $\theta(f)$ is on the fully isotropic (or lightlike) plane $I(v)$ containing $p(v)$ and spanned by the two direction vectors $(x(v),1,0)$ and $(y(v),0,1)$. The plane $I(v)$ can be understood as Lorentz lift of $p(v)$. Thus, the Lorentz lift is a polyhedral surface in $\R^{2,2}$ with planar facets such that each facet is fully isotropic.

Conversely, given such a polyhedral surface the (orthogonal coordinate) projection to $\R^2$ is a t-realization and the intersection of the fully isotropic planes with $\R^2$ yields the intersection points of the circle pattern, that is, $p(v) = I(v) \cap \R^2$ for all $v \in H$. Note that fully isotropic planes corresponding to adjacent vertices of $H$ are of different type, because the orientations of the frames $\Omega$ are different.

We also consider the quadric $\mathcal Q$ in $\RP^5$ defined in homogeneous coordinates by
\begin{align}
    \mathcal Q = \{ [P] \in \RP^5 \mid P_1^2 + P_2^2 + P_3^3 - P_4^2 - P_5^2 - P_6^2 = 0 \}.
\end{align}
The group of projective transformations preserving $\mathcal Q$ is denoted by $\mathrm{PO}(3,3)$. For stereographic projection, we choose affine coordinates such that $P_6 = 1$, so that the affine image of the quadric in $\R^5$ is given by
\begin{align}
    \mathcal Q\cap \R^5 = \{ [P] \in \RP^5 \mid P_1^2 + P_2^2 + P_3^3 - P_4^2 - P_5^2 = 1 \}.
\end{align}
For stereographic projection $\pi: \mathcal Q \cap \R^5 \rightarrow \R^{2,2}$ we choose the projection center as $C = (1,0,0,0,0)$ and identify the codomain with $\R^{2,2} \simeq (0, P_2,P_3,P_4,P_5)$. Stereographic projection maps fully isotropic planes in $\mathcal Q$ (planes contained in $\mathcal Q$) to fully isotropic planes in $\R^{2,2}$ (planes containing only lightlike direction vectors) and vice versa.

Given a point $P \in \mathcal Q \subset \RP^n$, consider the set of all fully isotropic planes of one type through $P$. These planes lie in $\mathcal Q$ and the polar complement $P^\perp$ of $P$. The intersection $\mathcal Q \cap P^\perp$ is is a quadric of signature $\mathtt{(2,2,1)}$, that is, a cone over a one-sheeted hyperboloid $\mathcal H$. Each isotropic plane is obtained by joining an isotropic line of one type of $\mathcal H$ with $P$. As is well known, four isotropic lines in a one-sheeted hyperboloid have a projective invariant, which can be calculated as the cross-ratio of the intersections of the four lines with a generic plane $\mathcal E$ (using the conic that is the intersection of $\mathcal H$ with $\mathcal E$). Thus, four fully isotropic planes of one type passing through a common point inherit the cross-ratio as a projective invariant.

\begin{theorem}
    The face weights $X(f)$ of a circle pattern are invariant under projective transformations of the quadric $\mathcal Q$, that is, under a $\mathrm{PO}(3,3)$ action.
\end{theorem}
\proof{
    Consider first a face $f$ of degree 4. In this case, $X(f)$ is a cross-ratio of intersection points on a circle due to Lemma~\ref{lem:facemr}. On the other hand, this cross-ratio is precisely the invariant of four fully isotropic planes of the same type as discussed above (where $\R^2$ takes the role of $\mathcal E$), which proves the claim. If the degree of $f$ is larger than 4, we inductively reduce the degree as follows. Since $X(f)$ only depends on the local geometry at $f$ as in Figure~\ref{fig:facelabels} (left), we may forget the remainder of the graph. This means we can perform local resplits of the graph $G$ at one of the degree  white boundary vertices. As previously shown, this operation does not change $X(f)$, but it does reduce the degree of the face by 2.\qed
}

\section{Reductions}
\label{sec:reductions}

\begin{figure}
    \centering
     \begin{tikzpicture}[scale=1.4,baseline=(current bounding box.center)]
        \coordinate[wvert] (v1) at (1,0);
        \coordinate[wvert] (v3) at (1,1);
        \coordinate[wvert] (v2) at (1,2);
        \filldraw[fill opacity=0.25,fill=blue,draw=black,-]
            (v1.center) to[out=0,in=270] (2,1) to[out=90,in=0] (v2.center) -- cycle
            (v2.center) to[out=180,in=90] (0,1) to[out=270,in=180] (v1.center) -- cycle
        ;
        \coordinate[wvert,label=below:{$w_1$}] (v1) at (1,0);
        \coordinate[wvert,label={[xshift=1mm,yshift=0.5mm]below left:$w$}] (v3) at (1,1);
        \coordinate[wvert,label=above:{$w_2$}] (v2) at (1,2);

        \coordinate[bvert, blue,label={[xshift=-1mm,yshift=0mm]right:$b_2$}] (b1) at (1.6,1);
        \coordinate[bvert, blue,label={[xshift=1mm,yshift=0mm]left:$b_1$}] (b2) at (0.4,1);
        \draw[blue]
            (b1) edge (v1) edge (v2) edge (v3)
            (b2) edge (v1) edge (v2) edge (v3)
        ;
        \draw[-]
            (v3) edge node[above right] {$f_4$} (v1) edge node[below right] {$f_3$} (v2)
        ;
        \node[] at (1.6,0.5) {$f_2$};
        \node[] at (0.4,0.5) {$f_1$};
    \end{tikzpicture}
    \hspace{3mm}$\leftrightarrow$\hspace{1mm}
    \begin{tikzpicture}[scale=1.4,baseline=(current bounding box.center)]
       \coordinate[wvert,label=below:{$w_1$}] (v1) at (1,0);
        \coordinate[wvert,label=above:{$w_2$}] (v2) at (1,2);
        \draw[-]
             (v1)  -- (v2)
        ;
    \end{tikzpicture}
    \hspace{5mm}\vrule\hspace{6mm}
    \begin{tikzpicture}[scale=1.4,baseline=(current bounding box.center)]
        \coordinate[wvert] (v1) at (1,0);
        \coordinate[wvert] (v3) at (1,1);
        \coordinate[wvert] (v2) at (1,2);
        \filldraw[fill opacity=0.25,fill=red,draw=black,-]
            (v1.center) to[out=0,in=270] (2,1) to[out=90,in=0] (v2.center) -- cycle
            (v2.center) to[out=180,in=90] (0,1) to[out=270,in=180] (v1.center) -- cycle
        ;
        \coordinate[wvert,label=below:{$w_1$}] (v1) at (1,0);
        \coordinate[wvert,label={[xshift=1mm,yshift=0.5mm]below left:$w$}] (v3) at (1,1);
        \coordinate[wvert,label=above:{$w_2$}] (v2) at (1,2);

        \coordinate[bvert, red,label={[xshift=-1mm,yshift=0mm]right:$r_2$}] (b1) at (1.6,1);
        \coordinate[bvert, red,label={[xshift=1mm,yshift=0mm]left:$r_1$}] (b2) at (0.4,1);
        \draw[red]
            (b1) edge (v1) edge (v2) edge (v3)
            (b2) edge (v1) edge (v2) edge (v3)
        ;
    \end{tikzpicture}
    \hspace{3mm}$\leftrightarrow$\hspace{1mm}
    \begin{tikzpicture}[scale=1.4,baseline=(current bounding box.center)]
       \coordinate[wvert,label=below:{$w_1$}] (v1) at (1,0);
        \coordinate[wvert,label=above:{$w_2$}] (v2) at (1,2);
        \draw[-]
             (v1)  -- (v2)
        ;
    \end{tikzpicture}
    \hspace{5mm}\vrule\hspace{6mm}
    \begin{tikzpicture}[scale=1.4,baseline=(current bounding box.center)]
        \coordinate[wvert] (v1) at (1,0);
        \coordinate[wvert] (v3) at (1,1);
        \coordinate[wvert] (v2) at (1,2);
        \filldraw[fill opacity=0.25,fill=blue,draw=black,-]
            (v1.center) to[out=0,in=270] (2,1) to[out=90,in=0] (v2.center) -- cycle
        ;
        \filldraw[fill opacity=0.25,fill=red,draw=black,-]
            (v2.center) to[out=180,in=90] (0,1) to[out=270,in=180] (v1.center) -- cycle
        ;
        \coordinate[wvert,label=below:{$w_1$}] (v1) at (1,0);
        \coordinate[wvert,label={[xshift=1mm,yshift=0.5mm]below left:$w$}] (v3) at (1,1);
        \coordinate[wvert,label=above:{$w_2$}] (v2) at (1,2);

        \coordinate[bvert, blue,label={[xshift=-0.5mm,yshift=0mm]right:$b$}] (b) at (1.6,1);
        \draw[blue]
            (b) edge (v1) edge (v2) edge (v3)
        ;
        \coordinate[bvert, red,label={[xshift=0.5mm,yshift=0mm]left:$r$}] (r) at (0.4,1);
        \draw[red]
            (r) edge (v1) edge (v2) edge (v3)
        ;
        \draw[-]
            (v3) edge (v1) edge  (v2)
        ;
        \node[] at (1.6,0.5) {$f_2$};
        \node[] at (0.4,0.5) {$f_1$};
    \end{tikzpicture}
    \hspace{3mm}$\leftrightarrow$\hspace{1mm}
    \begin{tikzpicture}[scale=1.4,baseline=(current bounding box.center)]
       \coordinate[wvert,label=below:{$w_1$}] (v1) at (1,0);
        \coordinate[wvert,label=above:{$w_2$}] (v2) at (1,2);
        \draw[-]
             (v1)  -- (v2)
        ;
    \end{tikzpicture}

    \begin{tikzpicture}[line cap=round,line join=round,>=triangle 45,x=1.0cm,y=1.0cm,scale=0.4,baseline=(current bounding box.center)]
        \clip(2.32,-5.64) rectangle (9.36,1.48);
        \draw [thick, red] (5.774398596491228,-2.0514212865497083) circle (3.2922875115844032cm);
        \node[wvert, label=below:$w_1$] at (7.52,0.74) {};
        \node[wvert, label=right:$w$] at (2.8,-0.64) {};
        \node[wvert, label=above:$w_2$] at (5.38,-5.32) {};
    \end{tikzpicture}
    \hspace{20mm}
    \begin{tikzpicture}[line cap=round,line join=round,>=triangle 45,x=1.0cm,y=1.0cm,scale=0.4,baseline=(current bounding box.center)]
        \clip(2.32,-6) rectangle (13,1.48);
        \draw [thick, red] (5.774398596491228,-2.0514212865497083) circle (3.2922875115844032cm);
        \node[wvert, label=below:$w_1$] (w1) at (7.52,0.74) {};
        \node[wvert, label={[xshift=1mm]above:$w_2$}] (w2) at (5.38,-5.32) {};
        \node[bvert,red, label=right:$r_{1}{=}r_2$] (r12) at (8.83583274633734,-3.262527394655284) {};
        \node[wvert,label=left:$w$] (w) at (6,-2) {};
        \tkzDefCircle[circum](w1,w,r12) \tkzGetPoint{m1};
        \tkzDrawCircle[red,thick](m1,w);
        \tkzDefCircle[circum](w2,w,r12) \tkzGetPoint{m2};
        \tkzDrawCircle[red,thick](m2,w);
        \node[wvert] at (w) {};
        \node[wvert] at (w1) {};
        \node[wvert] at (w2) {};
        \node[bvert,red] at (r12) {};
    \end{tikzpicture}
    \hspace{20mm}
    \begin{tikzpicture}[line cap=round,line join=round, >=triangle 45, x=1.0cm, y=1.0cm, scale=0.4, baseline=(current bounding box.center)]
        \clip(2.32,-5.64) rectangle (9.36,1.48);
        \draw [thick, red] (5.774398596491228,-2.0514212865497083) circle (3.2922875115844032cm);
        \node[wvert, label=below:$w_1$] at (7.52,0.74) {};
        \node[wvert, label=right:$w$] at (2.8,-0.64) {};
        \node[wvert, label=above:$w_2$] at (5.38,-5.32) {};
        \node[bvert, red, label=left:$r$] at (8.83583274633734,-3.262527394655284) {};
    \end{tikzpicture}

    \caption{The three reduction moves from left to right: blue reduction, red reduction and blue-red reduction.}
    \label{fig:reduction}
\end{figure}
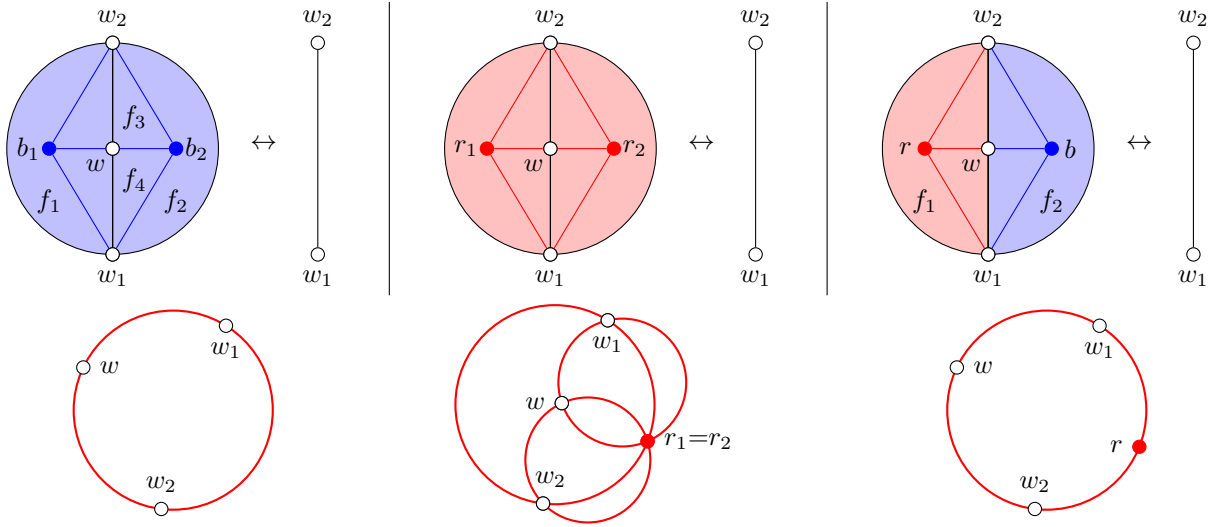

We consider three reduction moves for bicolored triangulations as depicted in Figure~\ref{fig:reduction}, each removing an interior degree 2 vertex and the incident edges and triangles and then gluing the two outer edges: \emph{blue reduction}, \emph{red reduction}, and \emph{blue-red reduction}.

Consider the geometry of red reduction. By definition we know that $p(r_1) = p(r_2)$. Hence, the two circles to the left and to the right share the three points $p(w_1)$, $p(w_2)$ and $p(r_1)$, thus, they coincide. Therefore we can perform the reduction and the resulting graph $H$ still represents the corresponding circle pattern with one point removed. Note that nothing about $G$ or the associated dimer weights changes. Therefore, even with information about the dimer weights on $G$ before, this move is not reversible as the information about $p(w)$ is lost. That said, if this move is possible then $G$ has a non-simply connected face, which implies that the circle pattern does not have a convex t-embedding.

Next we consider blue reduction. Removing $w$ corresponds to removing an isolated point from a circle, so this reduction is also geometrically well-defined. Let us denote the outer left face of $G$ by $f_1$, the outer right by $f_2$, the inner top by $f_3$, the inner bottom by $f_4$, and the new face glued from $f_1$ and $f_2$ we denote by $f$. Due to the manner in which the edge weights are determined by $p(w)$, $p(w_1)$ and $p(w_2)$ we see that $X(f_3) = X(f_4) = -1$, and if neither $f_1$ nor $f_2$ are a boundary face, we get $X(f) = X(f_1)X(f_2)$. If this move is possible then $H$ has a non-simply connected face, which implies that the circle pattern does not have a convex t-embedding.

Thirdly, consider blue-red reduction. Geometrically, the circles $p(f_1)$ and $p(f_2)$ coincide again because they share the three points $p(w_1)$, $p(w_2)$ and $p(r)$. Thus, this move corresponds to removing \emph{two} isolated points from a circle as well. This time $G$ is also affected, and we get $X(f) = X(f_1)X(f_2)$ in this case as well. Moreover, assuming we know the value of $X(f_1)$ (or $X(f_2)$), the move is half-reversible in the sense that the position of $p(w)$ is determined by the remainder of the circle pattern and $X(f_1)$, but the position of $p(r)$ is not.

Recall that in \cite{klrrdimers} it was proven that for a certain class of circle patterns with convex t-embedding (cyclic outer degree 4 face) with given boundary points and Euclidean face weights, there exist exactly two (possibly coinciding) realizations of the graph as circle pattern. The proof strategy relies on reduction moves as induction step, and then solving the problem for a small minimal graph as induction start.

We may try to repeat this strategy in the Möbius setup. Clearly, only blue-red reductions can occur for circle patterns with convex t-embedding. However, as stated above, these reductions are generally not reversible even with knowledge of the Möbius face weights. Hence, only statements about the number of realizations which have the same white vertex geometry survive the induction process, the geometric information about red vertices is lost. Consequently, we do not pursue this approach further in the present article.

\begin{remark}
    To some extent it is clear that the Möbius face weights do not contain enough information to make the reductions fully reversible. However, it is difficult to say exactly how much information is missing, since as soon as a face has degree more than 4, the white vertices actually do determine the corresponding faces and thus also some of the red points. It is also possible that one should associate a second set of Möbius face weights to a circle pattern by interchanging the role of black and white vertices. In the presence of a boundary it is currently not clear how to perform this interchange exactly, but without boundary there is no obstacle to the definitions (and we consider one such case in Section~\ref{sec:isothermicp}). That said, then one would need to understand how to perform reductions simultaneously in both Möbius dimer models.
\end{remark}

\section{Examples}
\label{sec:examples}

\subsection{Isoradial graphs}

A \emph{(bipartite) isoradial circle pattern} is a circle pattern $p: H \rightarrow \C$ such that every circle has radius 1. They have been subject of a considerable amount of investigation \cite{ksisoradial,kenyondirac,csisoradial,bdtisosurvey}, thus, in this section we take a brief look at their Möbius weights.
If an isoradial circle pattern is embedded, then $H$ is necessarily minimal \cite{bcdtisoimm}, which we assume in the following. Let $Q$ be the quad-graph associated to $H$, which has vertex set $H \cup H^*$ and an edge for every pair of incident vertex and face in $H$. A \emph{train track} is a maximal sequence of edges in $Q$ such that consecutive edges share a face but not a vertex. As is well known, to each train track one may associate an angle $\eta \in \S^1$ so that the images of the edge vectors in this train track are given by $\exp(i\eta)$. Consider a blue vertex $b\in G$ with neighbours $w_1,w_2,w_3 \in G$. Recall that $b$ corresponds to a face $f_b$ of $H$. Let $\eta_1$, $\eta_2$, $\eta_3$ be the three associated angles, so that $p(w_k) - t(f_b) = \exp(i\eta_k)$ for all three $k$. Clearly, the canonical edge-coefficients are of the form
\begin{align}
    \mu(w_k,b) = \exp(i\eta_{k+1}) - \exp(i\eta_{k-1}). \label{eq:isoedge}
\end{align}
To calculate a face weight $X(f)$ of a face of $G$, we take alternating ratios of the $\mu$ expressions (and a Kasteleyn sign). Since every train track angle appears twice, we may manipulate the expression to obtain
\begin{align}
    X(f) = \prod_{k=1}^m \frac{\sin \frac12(\eta_k - \eta_{k+1})}{\sin\frac12(\eta_{k+1}-\eta_{k+2})}. \label{eq:isofaceweights}
\end{align}
This formula is reminiscent of the formula for Euclidean face weights, which is also an alternating ratio of sines of train track angle differences -- but with different combinatorics. It would be interesting to see if any results of \cite{kenyondirac} can also be obtained in the Möbius case.

\begin{remark}
    There is a generalization of isoradial circle patterns called \emph{integrable circle patterns} \cite{bmsanalytic}. These are circle patterns such that the intersection angles also factor onto the train tracks, but where the radii need not be constant. The class of integrable circle patterns contains Möbius transformations of isoradial circle patterns, but is significantly larger.
    To calculate $X(f)$ we need to take ratios of the weights \eqref{eq:isoedge}, each ratio living at one circle of the pattern. The fact that we can use the train track angles for each such ratio also holds for integrable circle patterns, thus, we believe that Equation~\eqref{eq:isofaceweights} also holds for general integrable circle patterns. Hence, from the viewpoint of the Möbius face weights, isoradial and integrable circle patterns appear to be equivalent.
\end{remark}

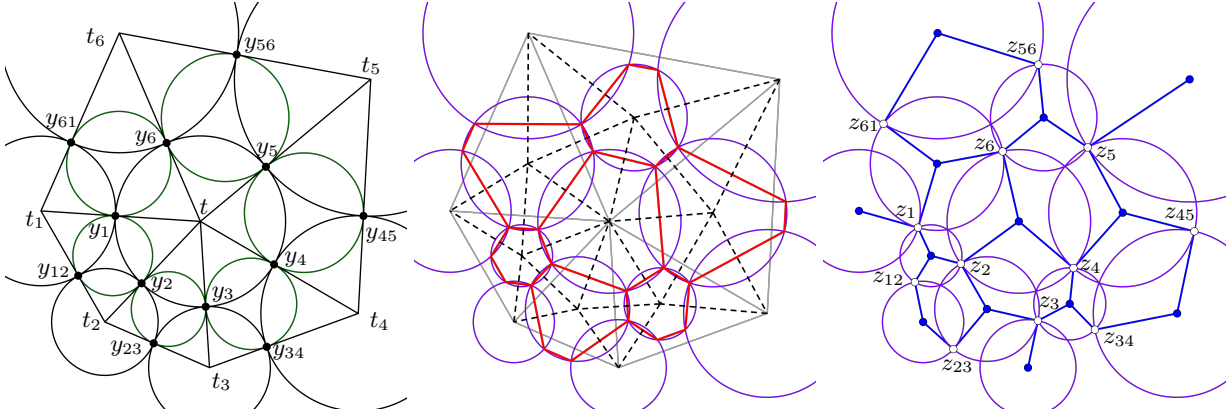
\begin{figure}
	\small
    \centering
    \begin{overpic}[width=0.32\linewidth, ,]{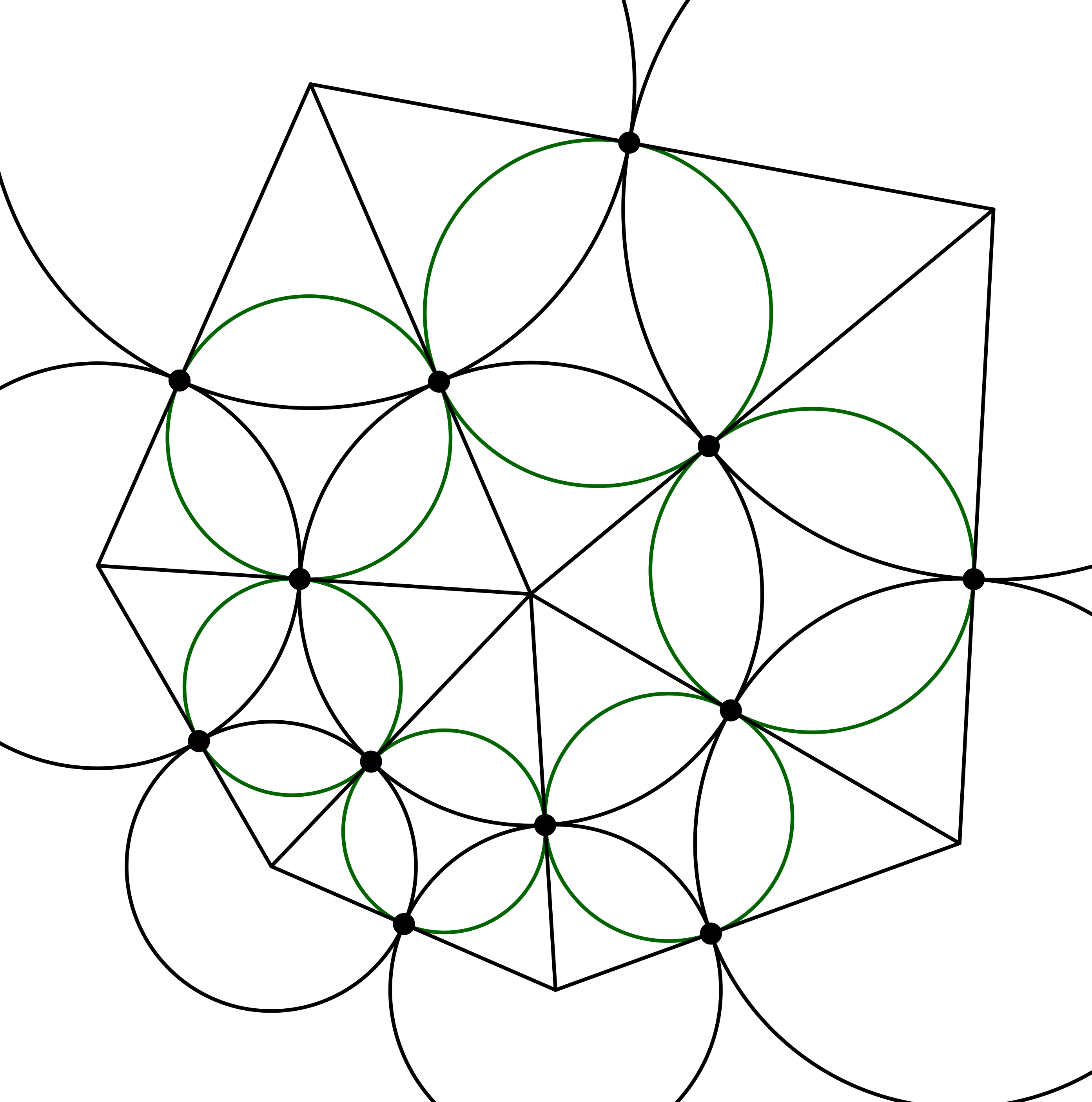}
    	\put(20.5,43){$y_1$}\put(36,30){$y_2$}\put(51,27.5){$y_3$}\put(69,36){$y_4$}\put(62,63){$y_5$}\put(32,66){$y_6$}
    	\put(8,33){$y_{12}$}\put(26,14){$y_{23}$}\put(66,13){$y_{34}$}\put(89,42){$y_{45}$}\put(59,89){$y_{56}$}\put(10,70){$y_{61}$}
    	\put(48,48){$t$}\put(5,45){$t_1$}\put(19,20){$t_2$}\put(51,5){$t_3$}\put(90,20){$t_4$}\put(88,83){$t_5$}\put(20,90){$t_6$}
    \end{overpic}
    \includegraphics[width=0.32\linewidth]{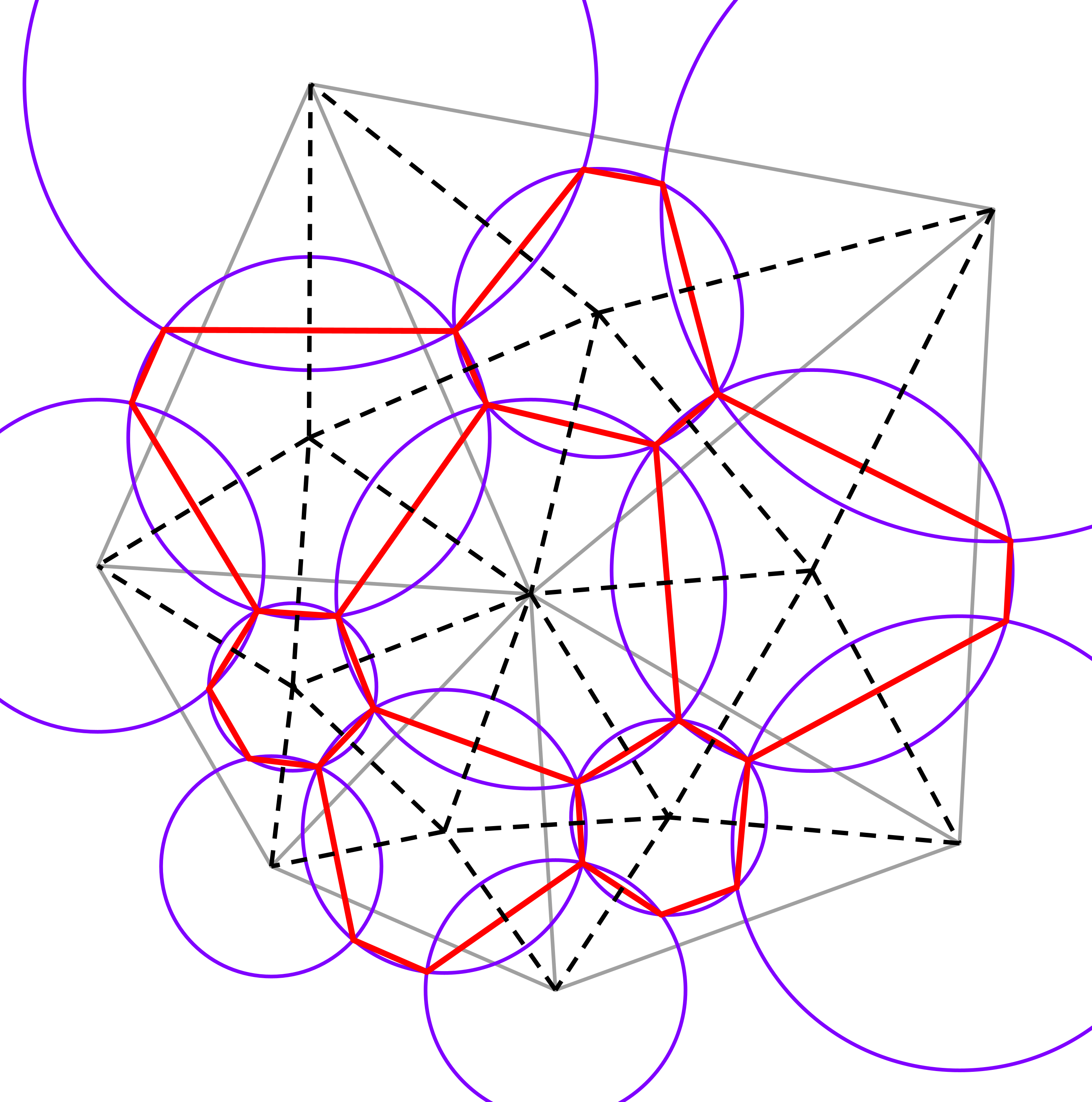}
    \begin{overpic}[width=0.32\linewidth, ,]{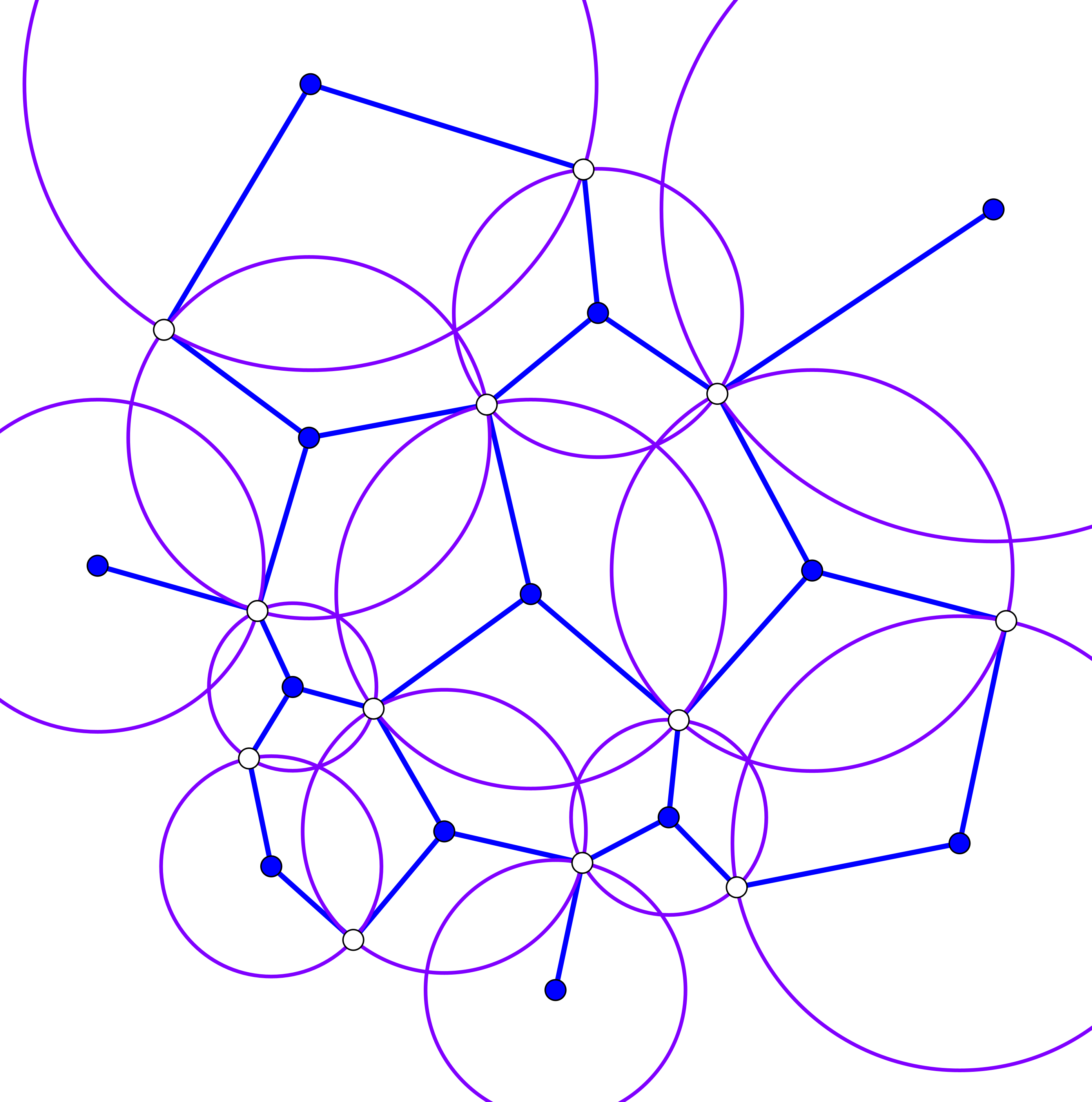}
		\put(17.5,47.5){$z_1$}\put(36,33.5){$z_2$}\put(53,25){$z_3$}\put(63,34){$z_4$}\put(67,62){$z_5$}\put(37,64){$z_6$}	\put(12,31){$z_{12}$}\put(29,9.5){$z_{23}$}\put(69,16){$z_{34}$}\put(84,47){$z_{45}$}\put(45,88){$z_{56}$}\put(6,69){$z_{61}$}
	\end{overpic}
    \caption{Left: a flower of a Doyle spiral (black) and face circles (green). Center: a circle pattern (violet) with the same t-embedding (dashed) as the combination of the Doyle spiral and the vertex circles, and the combinatorics of $H$ (red). Right: the corresponding dimer graph $G$ (blue).}
    \label{fig:doyle}
\end{figure}

\subsection{Doyle spirals}

A \emph{(local) Doyle spiral} is a hexagonal circle packing so that the circle radii around a circle of radius $r$ are $ra$, $rb$, $rb/a$, $r/a$, $r/b$, $ra/b$ for some positive numbers $a,b \in \R$ \cite{bdsdoyle}. A local configuration of seven such circles in a Doyle spiral is called a \emph{flower}. Not all choices of $a,b$ lead to a globally embedded circle packing, but this constraint is not relevant for our observations in the following.

To a Doyle spiral we also add the circles through the points of contact (green circles in Figure~\ref{fig:doyle}). By slightly changing the radii while keeping the t-embedding, we obtain a bipartite circle pattern (violet circles in  Figure~\ref{fig:doyle}, with combinatorics given by the red graph $H$). For each circle of the original Doyle spiral, there are three face-weights, for example
\begin{align}
    X(f) = \mr(z_1,z_{12},z_2,z_4,z_6,z_{61}),
\end{align}
with the labeling as in Figure~\ref{fig:doyle}. All other face weights in the Doyle spiral have the same three face weights, since there is a similarity taking any flower to any other flower.
Moreover, because the face-weights are independent of the choice of radii (while fixing the centers), we may express the same face-weight also by the points of contact of the original circle packing via
\begin{align}
    X(f) = \mr(y_1, y_{12}, y_2, y_4, y_6, y_{61}).
\end{align}
Moreover, it is well-known that the points of contact of a circle packing locally satisfy
\begin{align}
    \mr(y_1,y_2,y_{12},y_1,y_{61},y_6) = -1,
\end{align}
see \cite[Section 5.2]{ksclifford}. Therefore, one may simplify the expression for the face-weight to
\begin{align}
     X(f) = -\cro(y_1,y_2,y_4,y_6),
\end{align}
which is simply the cross-ratio of four contact points on a circle.
Moreover, in \cite[Prop.~2.1]{bhdoyle} it was shown that the six contact points on a circle satisfy
\begin{align}
	\mr(y_1,y_2,y_{3},y_4,y_{5},y_6) = -1.
\end{align}
As a result, the three face weights  at a circle of the Doyle spiral actually satisfy $X(f_1)X(f_2)X(f_3)=1$. Consequently, the two face-weights $X(f_1)$, $X(f_2)$ actually parameterize the space of (not necessarily globally coherent) Doyle spirals up to Möbius transformations.

\begin{remark}
	In principle, it is also possible to consider the Euclidean weights on $H$ associated to a Doyle spiral via distances in the t-embedding (as in \cite{amiquel,klrrdimers}). However, a brief calculation shows that the Euclidean face weights are all equal to
	\begin{align}
		Y = -\frac{(t_1-t)(t_3-t)(t_5-t)}{(t_2-t)(t_4-t)(t_6-t)} = \frac{(1+a)(1+b/a)(1+1/b)}{(1+b)(1+1/a)(1+a/b)} = 1.
	\end{align}
	Consequently, the Euclidean weights cannot distinguish different Doyle spirals, while the Möbius weights do distinguish different Doyle spirals up to Möbius transformations.
\end{remark}

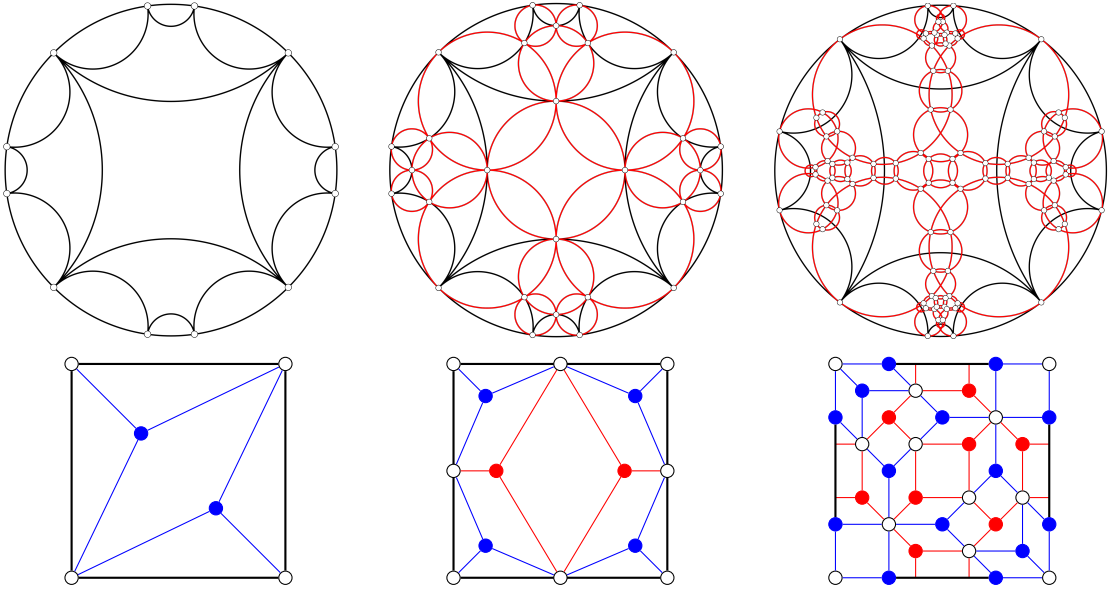
\begin{figure}[tb]
    \centering
	\includegraphics[width=0.27\linewidth]{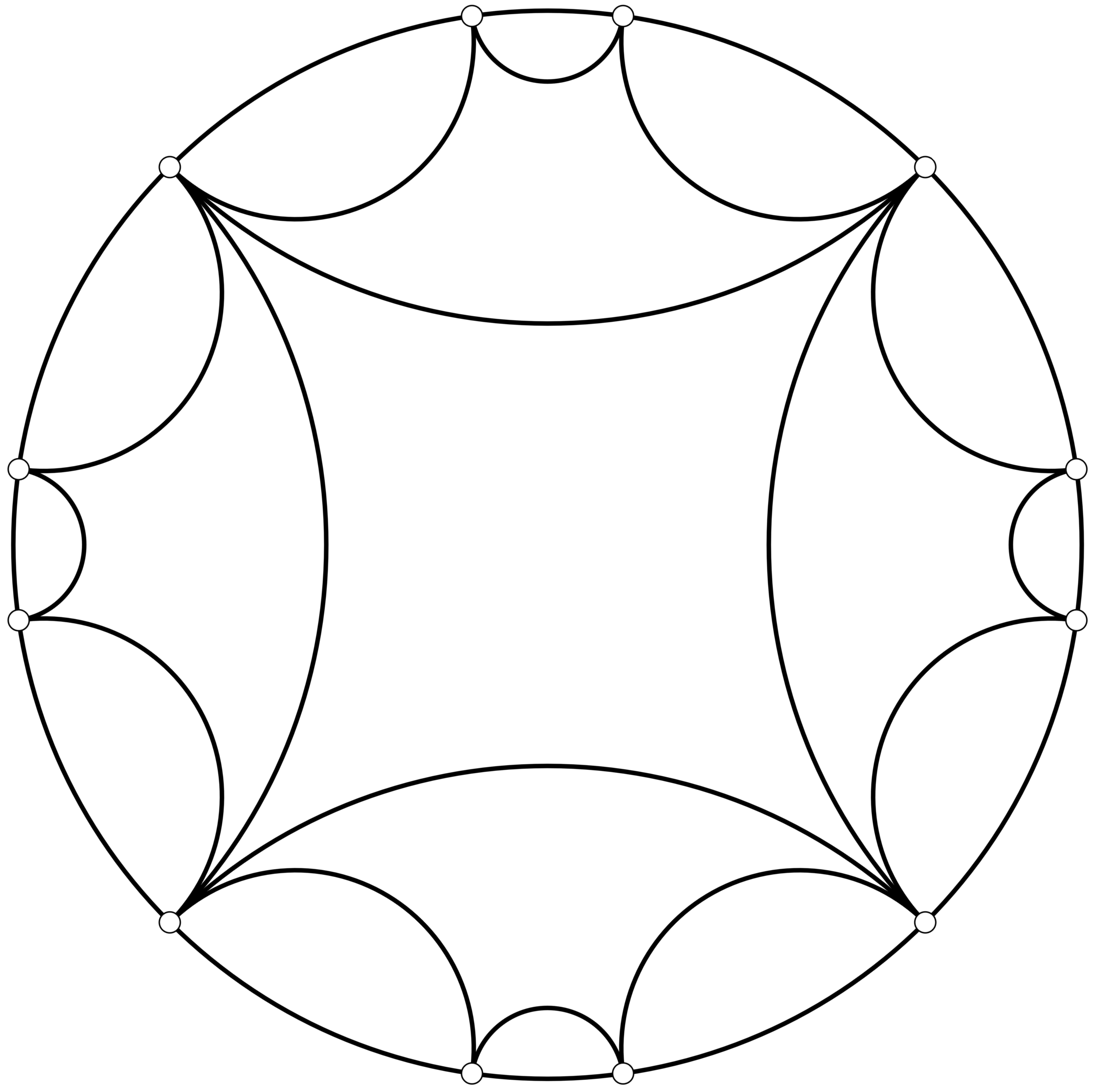}\hspace{5mm}
	\includegraphics[width=0.27\linewidth]{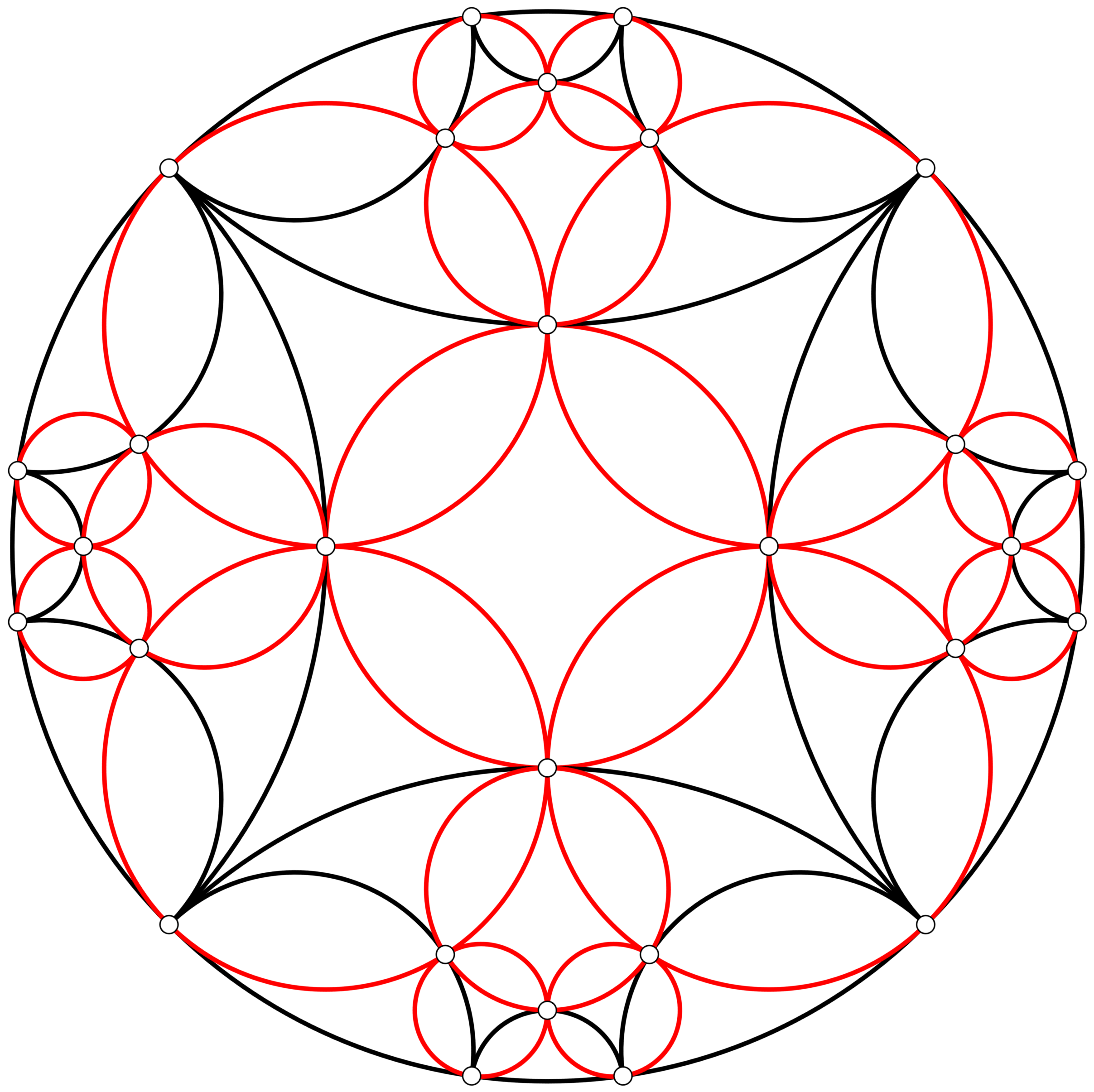}\hspace{5mm}
	\includegraphics[width=0.27\linewidth]{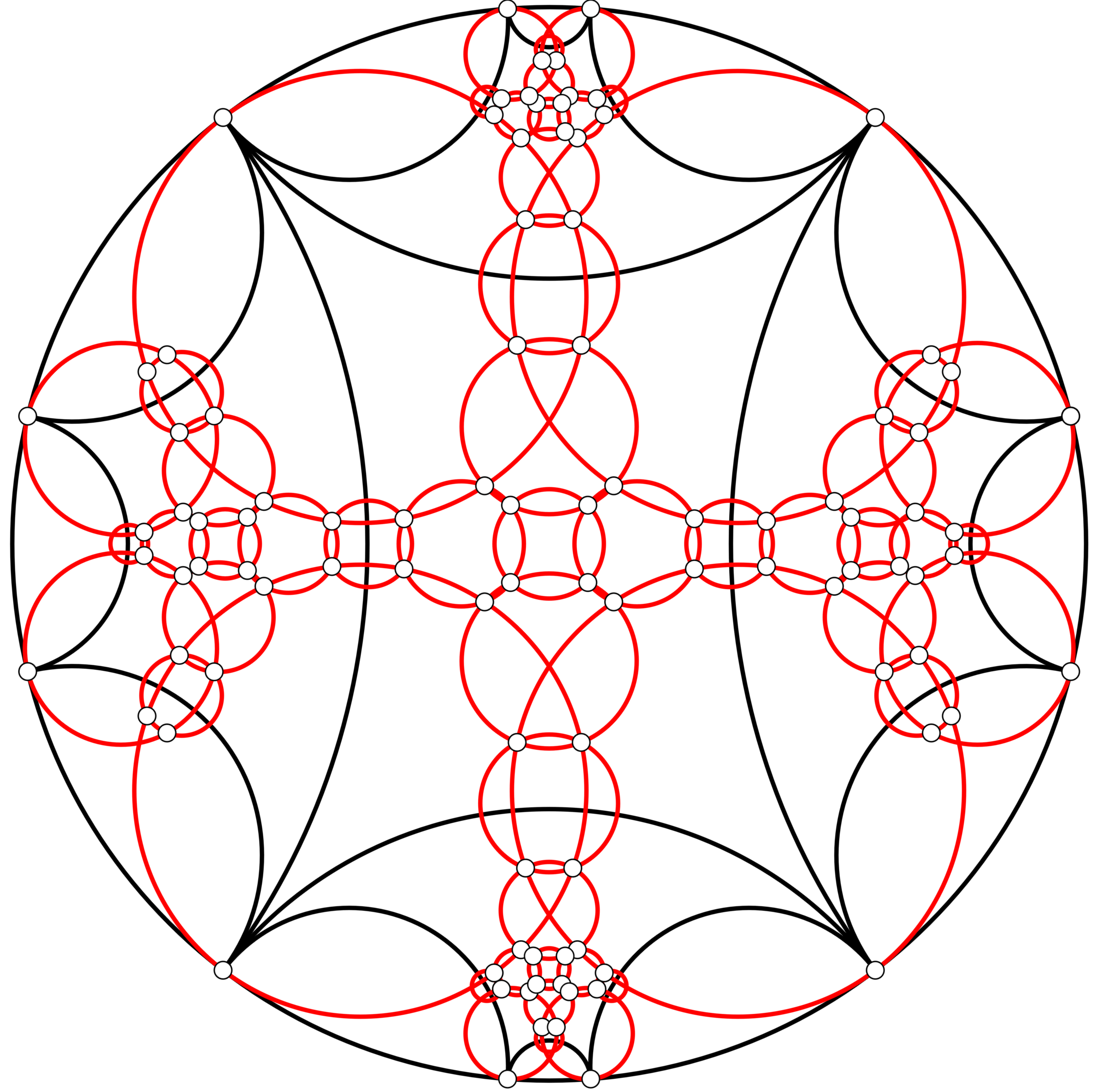}
    \\ \vspace{2mm}
    \hspace{0mm}
    \begin{tikzpicture}[scale=2, rotate=45]
        \node[wvert] (ws) at (0,-1) {};
        \node[wvert] (we) at (1,0) {};
        \node[wvert] (ww) at (-1,0) {};
        \node[wvert] (wn) at (0,1) {};
        \draw[-, thick]
            (ws) -- (we) -- (wn) -- (ww) -- (ws)
        ;
        \node[bvert,blue] (bs) at (0,-0.35) {};
        \node[bvert,blue] (bn) at (0,0.35) {};
        \draw[-,blue]
            (bs) edge (ws) edge (we) edge (ww)
            (bn) edge (wn) edge (we) edge (ww)
        ;
    \end{tikzpicture}
    \hspace{18mm}
    \begin{tikzpicture}[scale=2, rotate=45]
        \node[wvert] (ws) at (0,-1) {};
        \node[wvert] (we) at (1,0) {};
        \node[wvert] (ww) at (-1,0) {};
        \node[wvert] (wn) at (0,1) {};
        \draw[-, thick]
            (ws) -- (we) -- (wn) -- (ww) -- (ws)
        ;
        \node[wvert] (wne) at (0.5,0.5) {};
        \node[wvert] (wse) at (0.5,-0.5) {};
        \node[wvert] (wnw) at (-0.5,0.5) {};
        \node[wvert] (wsw) at (-0.5,-0.5) {};
        \node[bvert,red] (rnw) at (-0.3,0.3) {};
        \node[bvert,red] (rse) at (0.3,-0.3) {};
        \draw[-,red]
            (rnw) edge (wnw) edge (wne) edge (wsw)
            (rse) edge (wse) edge (wne) edge (wsw)
        ;
        \node[bvert,blue] (bw) at (-0.7,0) {};
        \node[bvert,blue] (bn) at (0,0.7) {};
        \node[bvert,blue] (be) at (0.7,0) {};
        \node[bvert,blue] (bs) at (0,-0.7) {};
        \draw[-,blue]
            (bn) edge (wn) edge (wne) edge (wnw)
            (bs) edge (ws) edge (wse) edge (wsw)
            (bw) edge (ww) edge (wsw) edge (wnw)
            (be) edge (we) edge (wne) edge (wse)
        ;
    \end{tikzpicture}
    \hspace{18mm}
    \begin{tikzpicture}[scale=2, rotate=45]
        \node[wvert] (ws) at (0,-1) {};
        \node[wvert] (we) at (1,0) {};
        \node[wvert] (ww) at (-1,0) {};
        \node[wvert] (wn) at (0,1) {};
        \node[wvert] (wne) at (0.25,0.5) {};
        \node[wvert] (wse) at (0.25,-0.5) {};
        \node[wvert] (wnw) at (-0.25,0.5) {};
        \node[wvert] (wsw) at (-0.25,-0.5) {};
        \node[wvert] (wcn) at (0,0.25) {};
        \node[wvert] (wcs) at (0,-0.25) {};
        \node[wvert] (wcw) at (-0.5,0) {};
        \node[wvert] (wce) at (0.5,0) {};
        \node[bvert,red] (rn) at (0,0.5) {};
        \node[bvert,red] (rs) at (0,-0.5) {};
        \node[bvert,red] (rw) at (-0.25,0) {};
        \node[bvert,red] (re) at (0.25,0) {};
        \node[bvert,red] (rne) at (0.5,0.25) {};
        \node[bvert,red] (rse) at (0.5,-0.25) {};
        \node[bvert,red] (rnw) at (-0.5,0.25) {};
        \node[bvert,red] (rsw) at (-0.5,-0.25) {};
        \draw[-,red]
            (rn) edge (wnw) edge (wne) edge (wcn)
            (rs) edge (wsw) edge (wse) edge (wcs)
            (rw) edge (wcw) edge (wcn) edge (wcs)
            (re) edge (wce) edge (wcn) edge (wcs)
            (rne) edge (wce) edge (wne) edge ($(wn)!(rne)!(we)$)
            (wne) edge ($(wn)!(wne)!(we)$)
            (rnw) edge (wcw) edge (wnw) edge ($(wn)!(rnw)!(ww)$)
            (wnw) edge ($(wn)!(wnw)!(ww)$)
            (rse) edge (wce) edge (wse) edge ($(ws)!(rse)!(we)$)
            (wse) edge ($(ws)!(wse)!(we)$)
            (rsw) edge (wcw) edge (wsw) edge ($(ws)!(rsw)!(ww)$)
            (wsw) edge ($(ws)!(wsw)!(ww)$)
        ;
        \node[bvert,blue] (bn) at (0,0.75) {};
        \node[bvert,blue] (bs) at (0,-0.75) {};
        \node[bvert,blue] (bnw) at (-0.25,0.25) {};
        \node[bvert,blue] (bne) at (0.25,0.25) {};
        \node[bvert,blue] (bsw) at (-0.25,-0.25) {};
        \node[bvert,blue] (bse) at (0.25,-0.25) {};
        \node[bvert,blue] (bnnw) at (-0.25,0.75) {};
        \node[bvert,blue] (bnne) at (0.25,0.75) {};
        \node[bvert,blue] (bssw) at (-0.25,-0.75) {};
        \node[bvert,blue] (bsse) at (0.25,-0.75) {};
        \node[bvert,blue] (bwnw) at (-0.75,0.25) {};
        \node[bvert,blue] (bwsw) at (-0.75,-0.25) {};
        \node[bvert,blue] (bene) at (0.75,0.25) {};
        \node[bvert,blue] (bese) at (0.75,-0.25) {};
        \draw[-,blue]
            (bn) edge (wn) edge (wne) edge (wnw)
            (bs) edge (ws) edge (wse) edge (wsw)
            (bnw) edge (wcn) edge (wnw) edge (wcw)
            (bsw) edge (wcs) edge (wsw) edge (wcw)
            (bne) edge (wcn) edge (wne) edge (wce)
            (bse) edge (wcs) edge (wse) edge (wce)
            (bnnw) edge (wn) edge (wnw)
            (bnne) edge (wn) edge (wne)
            (bssw) edge (ws) edge (wsw)
            (bsse) edge (ws) edge (wse)
            (bwsw) edge (ww) edge (wcw)
            (bwnw) edge (ww) edge (wcw)
            (bese) edge (we) edge (wce)
            (bene) edge (we) edge (wce)
        ;
        \draw[-, thick]
            (bwnw) -- (bnnw)
            (bwsw) -- (bssw)
            (bene) -- (bnne)
            (bese) -- (bsse)
        ;
    \end{tikzpicture}
	\caption{Top row: three different circle patterns (red) on once-punctured torii with fundamental domains (black). Bottom: a fundamental domain of the combinatorics of the circle pattern $H$ (red) and corresponding dimer graph $G$ (blue).}
	\label{fig:ptorus}
\end{figure}

\subsection{Once-punctured torus}

Consider a triangulation $T$ of the disk without interior vertices and with only blue triangles. In this case the graph $H$ has no edges, and the circle pattern consists of a single circle with points on it. The interior faces of $G$ are all quads and the face weights $X(f)$ parametrize the geometry up to Möbius transformations.

This setup (blue triangles only) becomes more interesting if the topology is a higher genus surface instead of a disk. Indeed, we may now interpret each face weights $X(f)$ as a \emph{shear coordinate} of the corresponding edge in $T$ in the sense of punctured Teichmüller theory \cite{pennertm}. The boundary vertices correspond to the punctures on the ideal boundary, and the ideal boundary itself is the circle corresponding to the single face of $H$. In this sense, the one-colored version of our setup recovers the classical description of punctured Teichmüller theory.

Of course, one can wonder if it is not possible to generalize some part of punctured Teichmüller theory to the case where not all the triangles are blue, so that we may consider non-trivial circle patterns on higher genus surfaces. Indeed, we have illustrated a few examples in Figure~\ref{fig:ptorus}, where some of the boundary vertices correspond to punctures, the circles through the punctures are horocycles and the interior faces correspond to true hyperbolic circles. Ultimately, the question is whether there is a bijection between conformal class of a punctured surface with a choice of circle pattern, and a triangulation $T$ with choice of dimer weights on $G$ up to gauge, and we plan to investigate this in future research.

\subsection{Isothermic circle patterns}\label{sec:isothermicp}

In our setup red and white vertices of $H$ play different roles, since only the white vertices can be boundary vertices, and only the images $p(w)$ of white vertices are used to determine the edge and face weights for $G$. However, in the case of no boundary, we may exchange the colors of red and white vertices of $H$, and then construct an alternative dimer graph $G'$ using the images $p(r)$ of red vertices to determine the dimer weights.
In general, it is not clear how the two sets of weights on $G$ and $G'$ are related.

In this section we consider the special case where $H$ and $H'$ are both equivalent to  $\Z^2$ up to contractions, which allows one to canonically identify the sets of faces of $G$ and $G'$, see also
Figure~\ref{fig:koenigssplit} for an illustration of the graphs $H$, $H'$, $G$ and $G'$. Note that the white vertices in $H$ and $H'$ are not mapped to the same points by $p$: if $w$ is of degree 4 (in $H$ or $H'$) then it corresponds to an intersection point of the original circle pattern with $\Z^2$ combinatorics, but if $w$ is of degree 2 then it is an alternative intersection point of two circles in the original circle pattern. A vertex is of degree 4 in $H$ if and only if it is of degree 2 in $H'$ and vice versa, hence, the two sets of dimer weights do not use the same geometric data. We write $X'$ for the face weights of $G'$. In the following, we give a geometric characterization of the case where $X(f)=X'(f)$ for all faces.

\begin{figure}
	\centering
    \begin{tikzpicture}[scale=0.8]
		\coordinate[wvert] (w00) at (0,0);
		\coordinate[bvert,red] (w20) at (2,0);
		\coordinate[wvert] (w40) at (4,0);
		\coordinate[bvert,red] (w60) at (6,0);
		\coordinate[bvert,red] (w02) at (0,2);
		\coordinate[wvert,label=below left:$p_1$] (w22) at (2,2);
		\coordinate[bvert,red,label=below right:$p_2$] (w42) at (4,2);
		\coordinate[wvert] (w62) at (6,2);
		\coordinate[wvert] (w04) at (0,4);
		\coordinate[bvert,red,label=above left:$p_4$] (w24) at (2,4);
		\coordinate[wvert,label=above right:$p_3$] (w44) at (4,4);
		\coordinate[bvert,red] (w64) at (6,4);
		\coordinate[bvert,red] (w06) at (0,6);
		\coordinate[wvert] (w26) at (2,6);
		\coordinate[bvert,red] (w46) at (4,6);
		\coordinate[wvert] (w66) at (6,6);
		\draw[red,-]
			(w00) -- (w20) -- (w40) -- (w60) -- (w62) -- (w64) -- (w66) -- (w46) -- (w26) -- (w06) -- (w04) -- (w02) -- (w00)
			(w20) -- (w22) -- (w24) -- (w26)
			(w40) -- (w42) -- (w44) -- (w46)
			(w02) -- (w22) -- (w42) -- (w62)
			(w04) -- (w24) -- (w44) -- (w64)
		;
        \node at (3,3) {$f$};
	\end{tikzpicture}\hspace{4mm}
    \begin{tikzpicture}[scale=0.8,rotate=90]
		\coordinate[wvert] (w00) at (0,0);
		\coordinate[wvert] (w20) at (2,0);
		\coordinate[wvert] (w40) at (4,0);
		\coordinate[wvert] (w60) at (6,0);
		\coordinate[wvert] (w02) at (0,2);
		\coordinate[wvert,label=right:$p'_2$] (w22) at (2,2);
		\coordinate[wvert,label=right:$p_3$] (w42) at (4,2);
		\coordinate[wvert] (w62) at (6,2);
		\coordinate[wvert] (w04) at (0,4);
		\coordinate[wvert,label=left:$p_1$] (w24) at (2,4);
		\coordinate[wvert,label=left:$p'_4$] (w44) at (4,4);
		\coordinate[wvert] (w64) at (6,4);
		\coordinate[wvert] (w06) at (0,6);
		\coordinate[wvert] (w26) at (2,6);
		\coordinate[wvert] (w46) at (4,6);
		\coordinate[wvert] (w66) at (6,6);

        \coordinate[bvert,blue] (b00) at (0.6,1.4);
		\coordinate[bvert,blue] (b04) at (0.6,5.4);
		\coordinate[bvert,blue] (b40) at (4.6,1.4);
		\coordinate[bvert,blue] (b44) at (4.6,5.4);
		\coordinate[bvert,blue] (b22) at (2.6,3.4);
		\coordinate[bvert,blue] (bb00) at (1.4,0.6);
		\coordinate[bvert,blue] (bb04) at (1.4,4.6);
		\coordinate[bvert,blue] (bb40) at (5.4,0.6);
		\coordinate[bvert,blue] (bb44) at (5.4,4.6);
		\coordinate[bvert,blue] (bb22) at (3.4,2.6);
		\draw[-,blue]
			(b00) edge (w00) edge (w02) edge (w22)
			(bb00) edge (w00) edge (w20) edge (w22)
			(b40) edge (w40) edge (w42) edge (w62)
			(bb40) edge (w40) edge (w60) edge (w62)
			(b04) edge (w04) edge (w06) edge (w26)
			(bb04) edge (w04) edge (w24) edge (w26)
			(b44) edge (w44) edge (w46) edge (w66)
			(bb44) edge (w44) edge (w64) edge (w66)
			(b22) edge (w22) edge (w24) edge (w44)
			(bb22) edge (w22) edge (w42) edge (w44)
		;
		\coordinate[bvert,red] (r02) at (0.6,3.4);
		\coordinate[bvert,red,label=right:$p_2$] (r20) at (2.6,1.4);
		\coordinate[bvert,red,label=left:$p_4$] (r42) at (4.6,3.4);
		\coordinate[bvert,red] (r24) at (2.6,5.4);
		\coordinate[bvert,red,label=right:$p_2$] (rr02) at (1.4,2.6);
		\coordinate[bvert,red] (rr20) at (3.4,0.6);
		\coordinate[bvert,red,label] (rr42) at (5.4,2.6);
		\coordinate[bvert,red,label=left:$p_4$] (rr24) at (3.4,4.6);
		\draw[red,-]
			(r02) edge (w02) edge (w04) edge (w24)
			(r20) edge (w20) edge (w22) edge (w42)
			(r42) edge (w42) edge (w44) edge (w64)
			(r24) edge (w24) edge (w26) edge (w46)
			(rr02) edge (w02) edge (w22) edge (w24)
			(rr20) edge (w20) edge (w40) edge (w42)
			(rr42) edge (w42) edge (w62) edge (w64)
			(rr24) edge (w24) edge (w44) edge (w46)
		;
        \node at (3,3) {$f$};
	\end{tikzpicture}
	\hspace{4mm}
    \begin{tikzpicture}[scale=0.8]
		\coordinate[wvert] (w00) at (0,0);
		\coordinate[wvert] (w20) at (2,0);
		\coordinate[wvert] (w40) at (4,0);
		\coordinate[wvert] (w60) at (6,0);
		\coordinate[wvert] (w02) at (0,2);
		\coordinate[wvert,label=left:$p_1'$] (w22) at (2,2);
		\coordinate[wvert,label=right:$p_2$] (w42) at (4,2);
		\coordinate[wvert] (w62) at (6,2);
		\coordinate[wvert] (w04) at (0,4);
		\coordinate[wvert,label=left:$p_4$] (w24) at (2,4);
		\coordinate[wvert,label=right:$p_3'$] (w44) at (4,4);
		\coordinate[wvert] (w64) at (6,4);
		\coordinate[wvert] (w06) at (0,6);
		\coordinate[wvert] (w26) at (2,6);
		\coordinate[wvert] (w46) at (4,6);
		\coordinate[wvert] (w66) at (6,6);

		\coordinate[bvert,blue] (b00) at (0.6,0.6);
		\coordinate[bvert,blue] (b04) at (0.6,4.6);
		\coordinate[bvert,blue] (b40) at (4.6,0.6);
		\coordinate[bvert,blue] (b44) at (4.6,4.6);
		\coordinate[bvert,blue] (b22) at (2.6,2.6);
		\coordinate[bvert,blue] (bb00) at (1.4,1.4);
		\coordinate[bvert,blue] (bb04) at (1.4,5.4);
		\coordinate[bvert,blue] (bb40) at (5.4,1.4);
		\coordinate[bvert,blue] (bb44) at (5.4,5.4);
		\coordinate[bvert,blue] (bb22) at (3.4,3.4);
		\draw[-,blue]
			(b00) edge (w00) edge (w20) edge (w02)
			(bb00) edge (w22) edge (w20) edge (w02)
			(b40) edge (w40) edge (w60) edge (w42)
			(bb40) edge (w62) edge (w60) edge (w42)
			(b04) edge (w04) edge (w24) edge (w06)
			(bb04) edge (w26) edge (w24) edge (w06)
			(b44) edge (w44) edge (w64) edge (w46)
			(bb44) edge (w66) edge (w64) edge (w46)
			(b22) edge (w22) edge (w42) edge (w24)
			(bb22) edge (w44) edge (w42) edge (w24)
		;

		\coordinate[bvert,red] (r02) at (0.6,3.4);
		\coordinate[bvert,red,label=left:$p_1$] (r20) at (2.6,1.4);
		\coordinate[bvert,red,label=right:$p_3$] (r42) at (4.6,3.4);
		\coordinate[bvert,red] (r24) at (2.6,5.4);
		\coordinate[bvert,red,label=left:$p_1$] (rr02) at (1.4,2.6);
		\coordinate[bvert,red] (rr20) at (3.4,0.6);
		\coordinate[bvert,red] (rr42) at (5.4,2.6);
		\coordinate[bvert,red,label=right:$p_3$] (rr24) at (3.4,4.6);
		\draw[red,-]
		(r02) edge (w02) edge (w04) edge (w24)
		(r20) edge (w20) edge (w22) edge (w42)
		(r42) edge (w42) edge (w44) edge (w64)
		(r24) edge (w24) edge (w26) edge (w46)
		(rr02) edge (w02) edge (w22) edge (w24)
		(rr20) edge (w20) edge (w40) edge (w42)
		(rr42) edge (w42) edge (w62) edge (w64)
		(rr24) edge (w24) edge (w44) edge (w46)
		;
        \node at (3,3) {$f$};
	\end{tikzpicture}

	\caption{Left: part of a $\Z^2$ without splitting. Center: after splitting red vertices we obtain $H$ (red), and a corresponding dimer graph $G$ (blue). Right: interchanging the coloring of $\Z^2$ first and subsequent splitting leads to a different graph $H'$ (red) with dimer graph $G'$ (blue).}
	\label{fig:koenigssplit}
\end{figure}
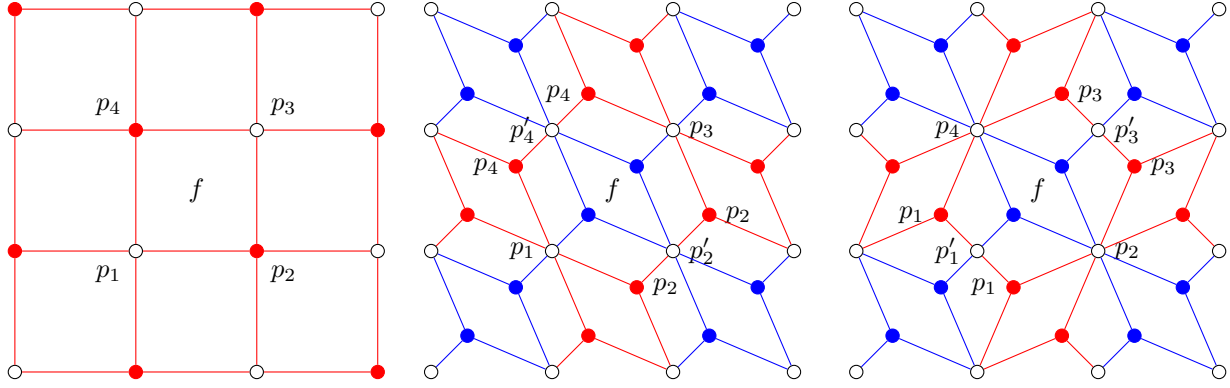

For the geometric characterization, we need the notion of a \emph{(discrete) K{\oe}nigs net} \cite{bsmoutard}, which play a prominent role in the theory of discrete differential geometry. K{\oe}nigs nets can be characterized in different ways, we use the following one which is based on the flatness of a discrete one-form \cite[Section~2.3]{ddgbook}, and allows us to define K{\oe}nigs nets in $\C$ as follows:
\begin{definition}\label{def:koenigs}
	A map $t: \Z^2 \rightarrow \C$ is a \emph{K{\oe}nigs net} if and only if for each vertex $v \in \Z^2$
	\begin{align}
		\mr(t(v_1), d(f_1),t(v_2),d(f_2),t(v_3),d(f_3),t(v_4),d(f_4)) = 1,\label{eq:koenigsmr}
	\end{align}
	where $v_i$ are the vertices adjacent to $v$ in cyclic order, $f_i$ is the face incident to $v,v_i,v_{i+1}$, and $d(f_i)$ is the intersection point of the diagonals of $t$ in face $f_i$.
\end{definition}

K{\oe}nigs nets are used to define several discretizations of isothermic nets, to us the following definition due to \cite{muellerconical} is relevant:

\begin{definition}
    A circle pattern is an \emph{isothermic circle pattern} if its t-embedding is a K{\oe}nigs net.
\end{definition}

With this notion, we can state the following theorem, which resolves a question the first author posed in an earlier publication \cite{amiquelmobius}.

\begin{theorem}
	The face weights on $G$ and $G'$ coincide if and only if the t-embedding is a K{\oe}nigs net.
\end{theorem}
\proof{
    Consider the local situation at a degree 4 face $f$ of $G$ (and $G'$), see Figure~\ref{fig:koenigssplit} for the combinatorics and Figure~\ref{fig:koenigsgeometry} for the local geometry.
    By invoking Lemma~\ref{lem:facemr} we immediately obtain that
    \begin{align}
        X(f) &= -\cro(p_1, p_2',p_3,p_4'), & X'(f) &= -\cro(p_1', p_2,p_3',p_4).
    \end{align}
    We denote the radius of the circle $p(f)$ by $r$. Then, by basic trigonometry we obtain that
    \begin{align}
        X(f) &= -\frac{\sin \omega_2 \sin \omega_4}{\sin \omega_2' \sin \omega_4'}, & X'(f) &= -\frac{\sin \omega_1' \sin \omega_3'}{\sin \omega_3 \sin \omega_1}.
    \end{align}
    Hence, the two face weights coincide if and only if
    \begin{align}
        \frac{\sin \omega_1 \sin \omega_2\sin \omega_3 \sin \omega_4}{\sin \omega_1' \sin \omega_2'\sin \omega_3' \sin \omega_4'} = 1,
    \end{align}
    which, by the law of sines, is equivalent to Equation~\eqref{eq:koenigsmr}, which proves the claim for $f$.

    If $f$ is instead a degree 8 face, then one may apply flips to both $T$ and $T'$ so that $f$ is of degree 4. Specifically, in $T$ we perform blue edge flips wherever possible followed by color flips, while in $T'$ one performs color flips followed by blue edge flips. The effect on the face weight of $f$ is $X(f) \mapsto X^{-1}(f)$ and $X'(f) \mapsto X'^{-1}(f)$ so that the same argument as above suffices.
    \qed
}

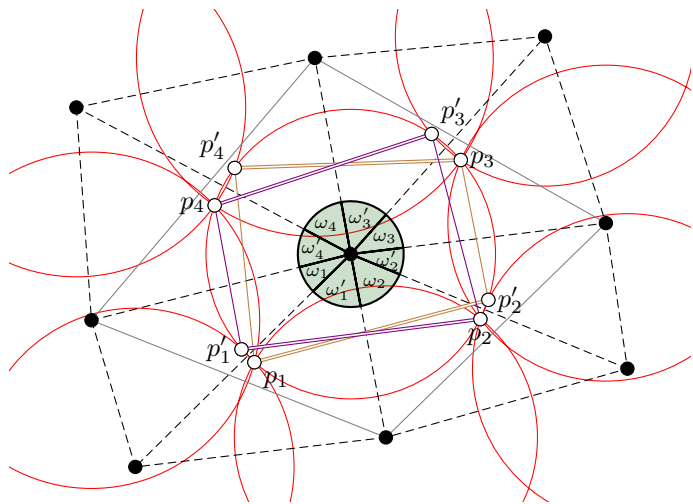
\begin{figure}
    \centering
    \begin{tikzpicture}[scale=0.5,line cap=round,line join=round,>=triangle 45,x=1cm,y=1cm]
        \clip (-4,-8.5) rectangle (14,4.5);
        \coordinate (A) at (-2.24,1.92);
        \coordinate (B) at (-1.84,-3.71);
        \coordinate (C) at (5.02,-1.96);
        \coordinate (D) at (4.07,3.23);
        \coordinate (E) at (10.16,3.80);
        \coordinate (F) at (11.77,-1.14);
        \coordinate (G) at (12.34,-4.99);
        \coordinate (H) at (5.95,-6.81);
        \coordinate (I) at (-0.68,-7.59);

        \coordinate[wvert] (p1) at (2.47,-4.82);
        \coordinate[wvert] (p2) at (8.45,-3.68);
        \coordinate[wvert] (p3) at (7.93,0.53);
        \coordinate[wvert] (p4) at (1.42,-0.67);
        \coordinate[wvert] (p1') at (2.13,-4.48);
        \coordinate[wvert] (p2') at (8.66,-3.17);
        \coordinate[wvert] (p3') at (7.16,1.22);
        \coordinate[wvert] (p4') at (1.95,0.32);

        \draw [red]
            (C) circle (3.83)
            (D) circle (4.72)
            (F) circle (4.18)
            (H) circle (4.01)
            (B) circle (4.45)
            (A) circle (4.48)
            (E) circle (3.96)
            (G) circle (4.10)
            (I) circle (4.20)
        ;
        \draw [draw=black, fill=darkgreen, fill opacity=0.2,thick]
            (C) -- ++(-165.74:1.40)
                arc[start angle=-165.74,end angle=-135.39,radius=1.4] -- cycle
            (C) -- ++(-135.39:1.4)
                arc[start angle=-135.39,end angle=-79.21,radius=1.4] -- cycle
            (C) -- ++(-79.21:1.4)
                arc[start angle=-79.21,end angle=-22.50,radius=1.4] -- cycle
            (C) -- ++(-22.50:1.4)
                arc[start angle=-22.50,end angle=6.96,radius=1.4] -- cycle
            (C) -- ++(6.96:1.4)
                arc[start angle=6.96,end angle=48.28,radius=1.4] -- cycle
            (C) -- ++(48.28:1.4)
                arc[start angle=48.28,end angle=100.43,radius=1.4] -- cycle
            (C) -- ++(100.43:1.4)
                arc[start angle=100.43,end angle=151.88,radius=1.4] -- cycle
            (C) -- ++(151.88:1.4)
                arc[start angle=151.88,end angle=194.26,radius=1.4] -- cycle;

        \draw [densely dashed, black]
            (A) -- (B) -- (C) -- (D) -- cycle
            (D) -- (E) -- (F) -- (G) -- (H) -- (C)
            (B) -- (I) -- (H)
            (C) -- (F)
            (C) -- (I)
            (C) -- (A)
            (C) -- (E)
            (C) -- (G)
        ;
        \draw [gray]
            (D) -- (B) -- (H) -- (F) -- cycle
        ;
        \draw [brown,double] (p1) -- (p2')  (p3) -- (p4');
        \draw [brown] (p2') -- (p3)  (p4') -- (p1);
        \draw [violet,double] (p1') -- (p2)  (p3') -- (p4);
        \draw [violet] (p2) -- (p3')  (p4) -- (p1');

        \node [right] at (p3) {$p_3$};
        \node [left] at (p4) {$p_4$};
        \node [below right] at (p1) {$p_1$};
        \node [below] at (p2) {$p_2$};
        \node [above left] at (p4') {$p_4'$};
        \node [above right] at (p3') {$p_3'$};
        \node [right] at (p2') {$p_2'$};
        \node [left] at (p1') {$p_1'$};

        \node [font=\scriptsize] at ($(C)+(210:1)$) {$\omega_1$};
        \node [font=\scriptsize] at ($(C)+(250:1)$) {$\omega_1'$};
        \node [font=\scriptsize] at ($(C)+(310:1)$) {$\omega_2$};
        \node [font=\scriptsize] at ($(C)+(350:1)$) {$\omega_2'$};
        \node [font=\scriptsize] at ($(C)+(25:1)$) {$\omega_3$};
        \node [font=\scriptsize] at ($(C)+(75:1)$) {$\omega_3'$};
        \node [font=\scriptsize] at ($(C)+(130:1)$) {$\omega_4$};
        \node [font=\scriptsize] at ($(C)+(171:1)$) {$\omega_4'$};

        \coordinate[bvert] (A) at (-2.24,1.92);
        \coordinate[bvert] (B) at (-1.84,-3.71);
        \coordinate[bvert] (C) at (5.02,-1.96);
        \coordinate[bvert] (D) at (4.07,3.23);
        \coordinate[bvert] (E) at (10.16,3.80);
        \coordinate[bvert] (F) at (11.77,-1.14);
        \coordinate[bvert] (G) at (12.34,-4.99);
        \coordinate[bvert] (H) at (5.95,-6.81);
        \coordinate[bvert] (I) at (-0.68,-7.59);
        \coordinate[wvert] (p1) at (2.47,-4.82);
        \coordinate[wvert] (p2) at (8.45,-3.68);
        \coordinate[wvert] (p3) at (7.93,0.53);
        \coordinate[wvert] (p4) at (1.42,-0.67);
        \coordinate[wvert] (p1') at (2.13,-4.48);
        \coordinate[wvert] (p2') at (8.66,-3.17);
        \coordinate[wvert] (p3') at (7.16,1.22);
        \coordinate[wvert] (p4') at (1.95,0.32);
    \end{tikzpicture}
    \caption{Local geometry at a face $f$ in an isothermic circle pattern: $X(f)$ corresponds to (minus) the brown cross ratio, $X(f')$ corresponds to the violet cross-ratio. }
    \label{fig:koenigsgeometry}
\end{figure}

In the $\Z^2$ case, one may also consider \emph{global Miquel dynamics}, every step of which consists of applying Miquel dynamics to every second circle -- alternating the choice of circle in every step.

\begin{corollary}
    The class of isothermic circle patterns is preserved by Miquel dynamics.
\end{corollary}
\proof{
    Since $X$ and $X'$ evolve under global Miquel dynamics by the same recursive formulas (cluster mutations / Y-system, see for example \cite{amiquelmobius}), the claim follows.\qed
}

\begin{remark}
	There are graphs with more general combinatorics than $\Z^2$, such that a canonical identification of faces of $G,G'$ is possible, but then the t-embedding is not necessarily a quad-graph. This makes it difficult to identify a meaningful characterization of the $X(f) = X'(f)$ case.
\end{remark}

\begin{remark}
	There are also special cases of graphs where one can identify the set of faces of $G$ with the set of faces of $H$, in this case it is possible to ask whether there are cases, such that the Möbius face weights on $G$ coincide with the Euclidean face weights on $H$ introduced in \cite{amiquel,klrrdimers}. One such special case where this happens are \emph{Tutte embeddings} \cite{tutteembedding}. This was already shown for the $\Z^2$ case in \cite{amiquelmobius}, and since the proof is local it immediately translates to more general combinatorics.
\end{remark}

\bibliographystyle{alphaurl}
\bibliography{references}

\end{document}